\documentclass[letterpaper]{article}
\usepackage{aaai}
\usepackage{times}
\usepackage{helvet}
\usepackage{courier}
\usepackage{url}
\usepackage{array}
\usepackage{amsthm}
\usepackage{amsmath}
\usepackage{amssymb}
\usepackage{hyperref}
\usepackage{enumitem}
\usepackage{booktabs}
\usepackage{multirow}
\usepackage{longtable}
\usepackage{algorithm}
\usepackage[table]{xcolor}
\usepackage{algpseudocodex}
\usepackage[most]{tcolorbox}

\newtheorem{remark}{Remark}
\newtheorem{proposition}{Proposition}
\DeclareMathOperator{\diag}{diag}
\algrenewcommand\algorithmicrequire{\textbf{Input:}}
\algrenewcommand\algorithmicensure{\textbf{Output:}}
\algrenewcommand\algorithmiccomment[1]{\hfill$\triangleright$~#1}
\def\MethodName{\textup{\textsc{REC-CBM}}}
\definecolor{RowGray}{gray}{0.93}
\tcbset{
    promptstyle/.style={
        colback=gray!5,
        colframe=black,
        fonttitle=\bfseries,
        boxrule=0.5mm,
        sharp corners,
        enhanced,
        breakable,
        width=\linewidth
    }
}

\begin{document}
%
\title{REC-CBM: Rubric-Aware Error-Correction Concept Bottleneck Models for Trustworthy Open-Ended Grading}
\author{AAAI Press\\
Association for the Advancement of Artificial Intelligence\\
2275 East Bayshore Road, Suite 160\\
Palo Alto, California 94303\\
}

\author{
    Chengshuai Zhao\textsuperscript{\rm 1}\thanks{Equal contribution.},
    Fan Zhang\textsuperscript{\rm 1}\footnotemark[1],
    Kumar Satvik Chaudhary\textsuperscript{\rm 1},
    Yiwen Li\textsuperscript{\rm 2},\\
    \Large\bf 
    Lo Pang-Yun Ting\textsuperscript{\rm 3},
    Ying-Chih Chen\textsuperscript{\rm 2},
    Huan Liu\textsuperscript{\rm 1}\\
    \textsuperscript{\rm 1}School of Computing and Augmented Intelligence, Arizona State University, USA\\
    \textsuperscript{\rm 2}Mary Lou Fulton Teachers College, Arizona State University, USA\\
    \textsuperscript{\rm 3}Department of Computer Science, National Yang Ming Chiao Tung University, TW\\
    {(czhao93; fzhan113; kchaud13; yiwenli2; ychen495; huanliu)}@asu.edu, lpyting@nycu.edu.tw\\
}

\maketitle
\begin{abstract}
Open-ended grading is central to equitable and personalized education, yet manual grading remains time-consuming and costly, underscoring the need for automated grading systems. Although recent neural and large language model (LLM) based systems have demonstrated superior performance, they are typically black-box models whose scoring processes and rationales are difficult for educators to verify and trust. Concept bottleneck models (CBMs) have emerged as a promising approach by routing predictions through human-interpretable concepts, providing a mechanistic guarantee of transparency. However, standard CBMs are not tailored to open-ended grading: they do not explicitly model fine-grained rubric dimensions, inadequately capture the ordinal semantics of scoring scales, and neglect inherent reliability issues in human concept annotations. To address these limitations, we propose \MethodName{}, a rubric-aware error-correction concept bottleneck model for trustworthy open-ended grading. \MethodName{} introduces a rubric-aware concept encoder that learns concept-specific representations over responses and an ordinal pairwise calibration objective that preserves ranking structure among rubric dimensions. It further incorporates a latent concept error-correction module that denoises concept predictions before final grade prediction while preserving interpretability. Comprehensive experiments on publicly available datasets show that \MethodName{} consistently improves grading performance and produces more faithful concept-level reasoning than both state-of-the-art black-box models and transparent baselines. Further analyses validate the contribution of each component and demonstrate the applicability in realistic educational settings. Overall, this work provides a practical, interpretable grading solution that enables educators to inspect, intervene in, and trust automated decisions, advancing more transparent and trustworthy education. Code is available:~\href{https://github.com/scott-f-zhang/REC-CBM}{https://github.com/scott-f-zhang/REC-CBM}. Data are available:~\href{https://huggingface.co/datasets/scott-f-zhang/REC-CBM}{https://huggingface.co/datasets/scott-f-zhang/REC-CBM}
\end{abstract}

\section{Introduction.}
Open-ended grading evaluates students' written responses to open questions, which plays a fundamental role in education across subjects and levels. At the same time, evaluating free-text responses at scale is labor-intensive and expensive, especially in settings that require timely feedback across large numbers of students~\cite{clauser2024automated}. This tension has made automated grading an important topic for educational data mining and AI for education communities~\cite{attali2006automated}.

\begin{figure}[t]
\center
\includegraphics[width=\columnwidth]{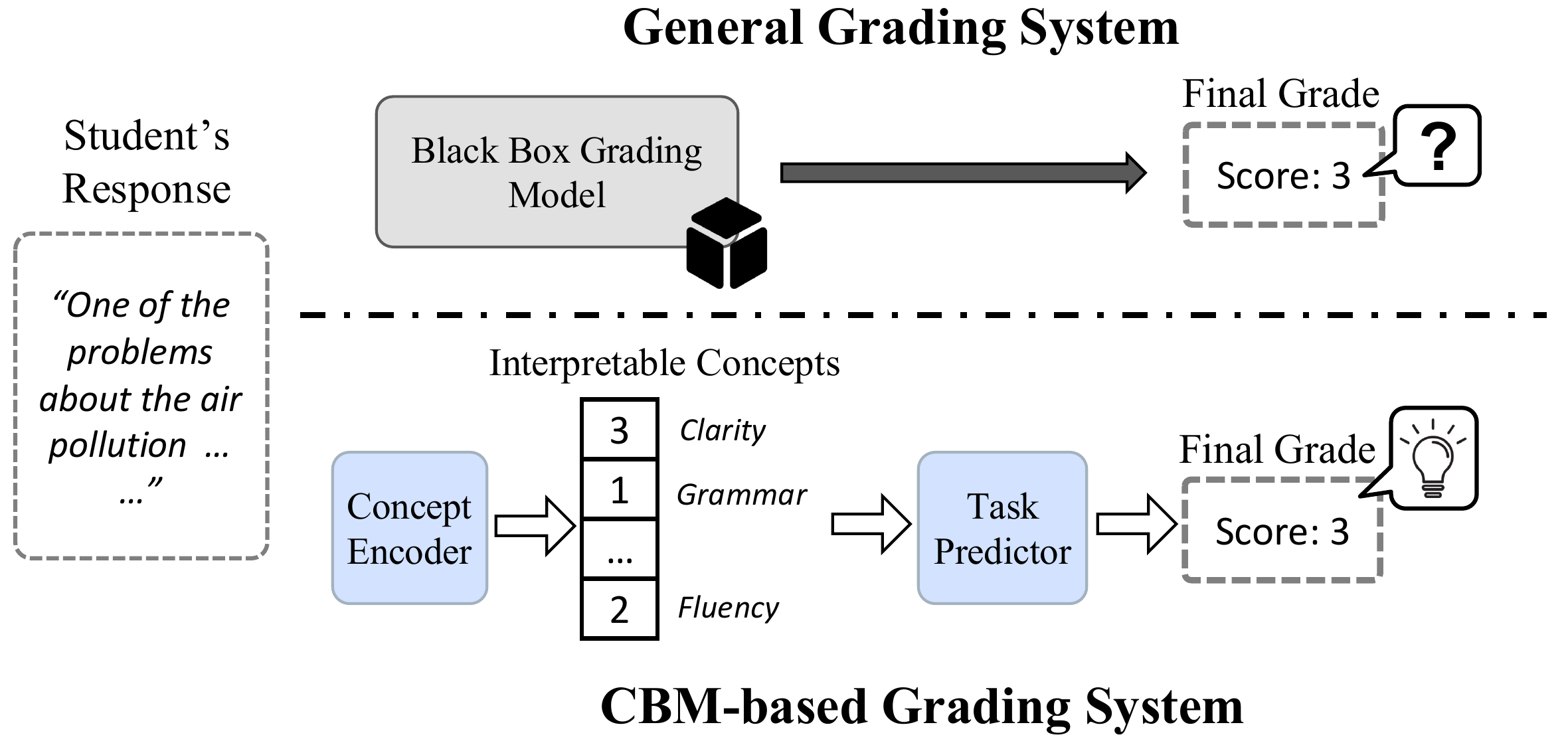}
\vspace{-6mm}
\caption{Comparison of open-ended grading systems.}
\vspace{-4mm}
\label{fig:illustration}
\end{figure}

Recent neural and large language model based grading systems have substantially improved predictive performance~\cite{wang2022use,golchin2025grading}, making automated scoring increasingly viable in practice. However, educational assessment is a high-stakes domain where raw accuracy alone is insufficient. When a system assigns a score to a student's response, instructors need to understand why that score was produced, whether it aligns with the intended rubric, and how to intervene when the model makes a questionable judgment. The black-box nature of these models makes it difficult for educators to verify the reasoning behind a grade, contest, or audit specific decisions, and provide feedback that is pedagogically useful~\cite{wang2024beyond,chu2025rationale,zhao2025is}. This lack of transparency can undermine trust and limit the potential of automated grading systems in reliable educational contexts~\cite{winkelmes2023introduction}.

Concept bottleneck models (CBMs) offer a promising middle ground between predictive strength and user-centric transparency~\cite{koh2020concept}. By constraining predictions to pass through human-interpretable intermediate concepts, CBMs can expose task-specific evidence behind a final grade rather than leaving the decision process hidden inside an end-to-end encoder, which makes them particularly appealing for educational assessment.

Despite this promise, existing CBMs are not yet well-suited to open-ended grading. Firstly, standard CBMs assume all concepts reside in a shared latent space, which is suboptimal for fine-grained grading settings where concepts capture different rubric dimensions and attend to different aspects of responses. In addition, they treat concept scores as evenly spaced categories, which discards the natural ranking information in the rubric annotations. Furthermore, rubric concept labels in real world are inherently noisy and subjective due to annotator disagreement, overlapping rubric dimensions, and ambiguity in open-ended responses~\cite{zhao2025scale}. Existing CBMs generally assume concept labels are reliable, which can degrade grading performance and undermine trust in the concept-level explanations.

To address these limitations, we propose \MethodName{}, a rubric-aware error-correction concept bottleneck model for trustworthy open-ended grading. \MethodName{} first performs rubric-aware concept extraction so that each grading dimension is grounded in evidence relevant to that aspect of the response. It then applies ordinal concept calibration to preserve ranking structure among rubric levels~\cite{bradley1952rank}, thus capturing the score semantics. Finally, it incorporates latent concept error correction to denoise intermediate concept predictions while preserving the interpretable bottleneck.

Through comprehensive experiments, \MethodName{} demonstrates superior grading performance with faithful concept-level reasoning than both state-of-the-art black-box models and transparent baselines. Further analyses validate the contribution of each component and demonstrate the applicability in realistic educational settings. Our contributions are fourfold:

\begin{itemize}[leftmargin=*,nosep,topsep=2pt]
  \item \textbf{Problem formulation:} We approach trustworthy open-ended grading through the lens of transparent decision-making, grounding the grading process on human-interpretable concepts, enabling educators to inspect, verify, and intervene on reasoning.
  \item \textbf{Data resources:} We curate and annotate public open-ended grading benchmarks with rubric-aligned concept labels, creating a valuable resource for future research on interpretable educational assessment.
  \item \textbf{Methodological innovation:} We identify key limitations of standard CBMs in the grading context and propose \MethodName{}, a novel framework that integrates rubric-aware concept encoding, ordinal pairwise calibration, and latent concept error correction to address these limitations.
  \item \textbf{Empirical validation:} We conduct extensive evaluations and experiments, offering insights into the performance, interpretability, and practical utility of \MethodName{} in real-world grading scenarios.
\end{itemize}

\section{Related Work.}
In this section, we review three lines of work relevant to our work.
\subsection{Automated Grading Systems.}
Automated grading has evolved from feature-engineered systems such as e-rater~\cite{attali2006automated} to neural scoring models~\cite{clauser2024automated}. Pre-trained language models further improved accuracy~\cite{wang2022use}, and recent LLM-based approaches incorporate rubric descriptions through prompting~\cite{golchin2025grading}, yet their internal reasoning remains opaque. A growing line of work seeks to address this: Wang et al.~\cite{wang2024beyond} and Chu et al.~\cite{chu2025rationale} align scoring with human rationales, while Schaller et al.~\cite{schaller2024fairness} audit fairness in essay scoring. However, these approaches surface post-hoc explanations without mechanistically guaranteeing that the explanation determines the assigned grade.

\subsection{Concept Bottleneck Models.}
Concept Bottleneck Models~\cite{koh2020concept} improve transparency by routing predictions through human-interpretable concepts. Extensions address prediction uncertainty via probabilistic~\cite{kim2023probabilistic}, stochastic~\cite{vandenhirtz2024stochastic}, and post-hoc~\cite{yuksekgonul2023posthoc} concept bottlenecks, though none model measurement error from a psychometric perspective or account for ordinal concept structure.
CBMs have also been adapted beyond vision: Tan et al.~\cite{tan2024interpreting} interpret pretrained language models via concept bottlenecks, Ismail et al.~\cite{ismail2025concept} target protein design, and Espinosa Zarlenga et al.~\cite{zarlenga2023tabcbm} extend CBMs to tabular data. These works demonstrate the versatility of the CBM approaches in different domains, but the application to open-ended grading remains underexplored.

\subsection{Trustworthy AI in Education.}
Trustworthy AI in education, spanning fairness, transparency, accountability, and reliability, has become a growing concern as AI systems are deployed at scale~\cite{holmes2022ethics,miao2021ai}. Research has documented algorithmic bias across demographic groups in student modeling and admissions~\cite{baker2022algorithmic}, while explainable AI methods have been adapted for intelligent tutoring and learning analytics~\cite{khosravi2022explainable}. Complementary efforts improve the reliability of AI-powered educational tools, such as ontology-aware retrieval for cybersecurity education~\cite{zhao2025ontology}, and regulatory bodies have begun codifying principles for AI in high-stakes assessment~\cite{williamsonprinciples}.
We extend this line of research by achieving trustworthy open-ended grading through concept bottleneck models that enable educators to inspect, verify, and intervene in the grading process.

\section{Preliminaries.}
This section formalizes open-ended grading and reviews concept bottleneck models, which serve as the foundation of our method.

\subsection{Open-Ended Grading Task Formulation.}
We cast open-ended grading in education as an \emph{ordinal} multi-class classification problem. Each instance consists of a question or prompt $q$, a student's free-text response $r$, and optional auxiliary grading context $a$, such as a reference answer or rubric description. We represent a grading instance as
\begin{equation}
  \mathbf{x} = (q,\; r,\; a),
  \label{eq:input}
\end{equation}
and let $\mathcal{X}$ denote the space of grading instances. The grading target is an ordinal score $y \in \mathcal{Y} = \{0,1,\ldots,S\}$, where $S$ denotes the maximum attainable score. Given a labeled dataset $\mathcal{D} = \{(\mathbf{x}^{(n)}, y^{(n)})\}_{n=1}^{N}$, the goal is to learn a grading function
\begin{equation}
  f \colon \mathcal{X} \to \mathcal{Y}, \qquad \hat{y} = f(\mathbf{x}),
  \label{eq:task_function}
\end{equation}
that minimizes the empirical risk
\begin{equation}
  \hat{f} \in \arg\min_{f \in \mathcal{F}}
  \frac{1}{N} \sum_{n=1}^{N}
  \ell\!\left(f(\mathbf{x}^{(n)}), y^{(n)}\right),
  \label{eq:erm}
\end{equation}
for a hypothesis class $\mathcal{F}$ and loss $\ell$.

\subsection{Concept Bottleneck Models.}
\label{sec:cbm}
CBMs~\cite{koh2020concept} improve model transparency by introducing an intermediate layer of human-interpretable \emph{concepts} $\mathbf{c}$ between the raw input $\mathbf{x}$ and the final prediction $\hat{y}$.

A standard CBM decomposes prediction into a concept encoder $\phi$ and a task predictor $\psi$. Given an input $\mathbf{x}$, the concept encoder first estimates:
\begin{equation}
  \phi\colon \mathcal{X} \to \mathbb{R}^{K}, \qquad \hat{\mathbf{c}} = \phi(\mathbf{x}),
  \label{eq:concept_pred}
\end{equation}
and the task predictor then maps the predicted concepts to the final score:
\begin{equation}
  \psi\colon \mathbb{R}^{K} \to \mathcal{Y}, \qquad \hat{y} = \psi(\hat{\mathbf{c}}).
  \label{eq:task_pred}
\end{equation}

The objective of training a CBM is to minimize a joint loss that includes both the concept prediction loss $\ell_{\mathrm{con}}$ and the task prediction loss $\ell_{\mathrm{task}}$:

\begin{equation}
  \resizebox{1\hsize}{!}{$
  \hat{f} \in \arg\min_{f} 
  \frac{1}{N} \sum_{n=1}^{N} \Bigl[
  \lambda_c \ell_{\mathrm{con}}\!\left(\hat{\mathbf{c}}^{(n)}, \mathbf{c}^{(n)}\right) + \lambda_t \ell_{\mathrm{task}}\!\left(\hat{y}^{(n)}, y^{(n)}\right)\Bigr],
  $}
\label{eq:cbm_loss}
\end{equation}
where $\mathbf{c}^{(n)}$ is the human-annotated concept for instance $n$, and $\lambda_c, \lambda_t \geq 0$ are hyperparameters that balance the two losses. Equivalently, the overall model factors as $f = \psi \circ \phi$, so all task-relevant information must pass through the concept layer before a final prediction is produced, which provides a mechanistic interpretability.

The CBMs offer several advantages in educational settings. First, it provides \emph{transparency}: the model's final grade can be decomposed into rubric-level concept predictions that are easier for instructors to inspect. Second, it supports \emph{intervention and auditing}: educators can examine or modify concept scores and observe how the overall prediction changes. Third, it promotes \emph{pedagogical usefulness}: concept-level outputs align naturally with feedback categories used in teaching and assessment. These properties make CBMs a promising framework for trustworthy open-ended grading.

\section{The Proposed REC-CBM Framework.}

\begin{figure*}[t]
    \raggedright
    \includegraphics[width=0.95\textwidth]{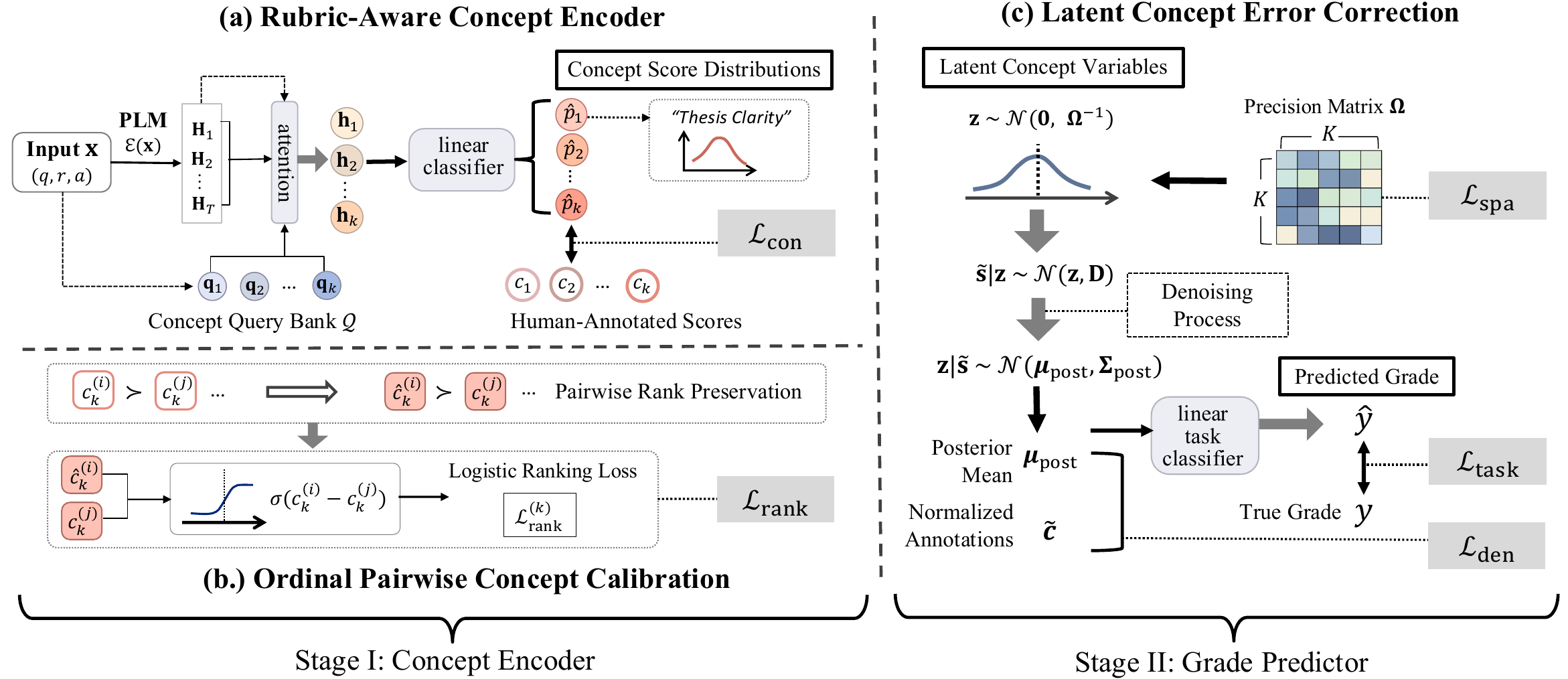}
    \caption{The proposed \MethodName{} framework.}
    \label{fig:framework}
\end{figure*}

\label{sec:method}
In this section, we present \MethodName{}, a rubric-aware error-correction concept bottleneck model as illustrated in Figure~\ref{fig:framework}. We first introduce the rubric-aware concept encoder (\S\ref{sec:race}), then detail the ordinal pairwise concept calibration (\S\ref{sec:ocr}), and the latent concept error correction mechanism (\S\ref{sec:lts}). Finally, we outline the combined learning objective and training paradigm (\S\ref{sec:training}).

\subsection{Rubric-Aware Concept Encoder.}
\label{sec:race}
Vanilla CBMs assume a shared representation space for all concepts, limiting fine-grained grading where rubric dimensions focus on different response aspects. We address this with a rubric-aware concept encoder.

Given an input $\mathbf{x} = (q, r, a)$, we employ a text encoder $\mathcal{E}$ (e.g., a pretrained language model) to obtain contextualized token representations:
\begin{equation}
  \mathbf{H} = \mathcal{E}(\mathbf{x}) \in \mathbb{R}^{T \times d},
  \label{eq:backbone}
\end{equation}
where $T$ is the sequence length and $d$ is the hidden dimension. Then, we introduce a \emph{concept query bank} to formulate the concept prototypes:
\begin{equation}
  \mathcal{Q} = \{\mathbf{q}_k\}_{k=1}^{K} \subset \mathbb{R}^d,
\end{equation}
where $\mathbf{q}_k$ is a learnable vector specializing in the textual evidence relevant to rubric dimension $k$, initialized via QR orthogonal decomposition of a random Gaussian matrix, promoting diverse patterns across concepts. 

For each concept $k$, we compute soft attention weights and aggregate the representation via:
\begin{align}
  \alpha_{k,t} &= \frac{\exp\!\left(\mathbf{q}_k^\top \mathbf{H}_t \;/\; \tau\right)}{\sum_{t'=1}^{T} \exp\!\left(\mathbf{q}_k^\top \mathbf{H}_{t'} \;/\; \tau\right)}, 
  \label{eq:attn} \\
  \mathbf{h}_k &= \sum_{t=1}^{T} \alpha_{k,t}\, \mathbf{H}_t,
  \label{eq:concept_repr}
\end{align}
where $\tau > 0$ is a temperature parameter. The mechanism selects rubric-specific token spans; each $\mathbf{q}_k$ acts as a dimension-specific retrieval query, rendering representations directly interpretable.


Let $M$ denote the maximum ordinal level of each rubric concept, so that $c_k \in \{0,\ldots,M\}$. We employ a linear classifier $\mathbf{V}_k \in \mathbb{R}^{(M+1) \times d}$ to predict the per-concept ordinal distribution $\hat{\mathbf{p}}_k$:
\begin{equation}
  \hat{\mathbf{p}}_k = \mathrm{Softmax}\!\left(\mathbf{V}_k \mathbf{h}_k\right) \in \mathbb{R}^{M+1},
\end{equation}
We define the concept prediction loss as the sum of cross-entropy over all concepts.
\begin{equation}
  \mathcal{L}_{\mathrm{con}} = \sum_{k=1}^{K} \mathrm{CrossEntropy}\!\left(\hat{\mathbf{p}}_k, c_k\right),
\end{equation}

\subsection{Ordinal Pairwise Concept Calibration.}
\label{sec:ocr}
The concept prediction loss from \S\ref{sec:race} treats the ordinal levels of each concept as unordered categories, which discards the natural ranking information in the rubric annotations. To preserve this ordinal structure, we introduce an ordinal pairwise concept calibration. 

First, we compute the expected ordinal concept score $\hat{c}_k$ from the predicted distribution $\hat{\mathbf{p}}_k$:
\begin{equation}
  \hat{c}_k = \mathbb{E}_{m \sim \hat{\mathbf{p}}_k}[m] = \sum_{m=0}^{M} m\,\hat{p}_{k,m} \in [0,M],
\end{equation}
so that $\hat{c}_k$ lives on the same ordinal scale $c_k$. For each concept, we collect the set of within-batch $B$ response pairs that admit a valid comparison under the rubric:
\begin{equation}
  \mathcal{P}_k = \bigl\{(i,j) : c_k^{(i)} > c_k^{(j)},\;
    i,j \in [B]\bigr\}.
  \label{eq:pair_set}
\end{equation}
Inspired by the Bradley-Terry model~\cite{bradley1952rank}, we minimize a logistic ranking loss over these pairs, which encourages the predicted scores to respect the observed ordering:
\begin{equation}
  \mathcal{L}_{\mathrm{rank}}^{(k)} = -\frac{1}{|\mathcal{P}_k|}
    \sum_{(i,j)\in\mathcal{P}_k}
    \log \sigma\!\left(\hat{c}_k^{(i)} - \hat{c}_k^{(j)}\right).
  \label{eq:loss_rank_k}
\end{equation}
where $\sigma(z) = 1/(1 + e^{-z})$ is the logistic sigmoid function. This loss encourages the model to assign higher predicted scores to responses that are annotated with higher concept levels, thus providing the ordinal relationship calibration in the rubric-aligned concepts.

The overall ranking loss averages over all concepts that admit at least one valid comparison in the batch:
\begin{equation}
  \mathcal{L}_{\mathrm{rank}} = \frac{1}{K'}
    \sum_{\substack{k=1 \\ |\mathcal{P}_k|>0}}^{K}
    \mathcal{L}_{\mathrm{rank}}^{(k)},
  \label{eq:loss_rank}
\end{equation}
where $K' = \bigl|\{k : |\mathcal{P}_k| > 0\}\bigr|$.


\subsection{Latent Concept Error Correction.}
\label{sec:lts}

Even after ordinal calibration, the predicted scores $\{\hat{c}_k\}$ from \S\ref{sec:ocr} still inherit annotator disagreement, overlap across rubric dimensions, and ambiguity in open-ended assessment. Passing these scores directly to the grading head would force the final predictor to rely on noisy concept estimates. We therefore introduce a latent correction layer that treats each calibrated concept score as an error-corrupted observation of an underlying latent concept state and then computes a closed-form posterior correction before grade prediction.

Let $\mathbf{z} \in \mathbb{R}^K$ denote the latent concept vector. We place a multivariate Gaussian prior on $\mathbf{z}$ with learnable precision matrix $\boldsymbol{\Omega} \in \mathbb{R}^{K \times K}$:
\begin{equation}
  \mathbf{z} \sim \mathcal{N}\!\left(\mathbf{0},\; \boldsymbol{\Omega}^{-1}\right).
  \label{eq:prior}
\end{equation}
we normalize each predicted concept to the unit interval,
\[
  \tilde{s}_k = \hat{c}_k / M \in [0,1],
\]
and model it as a noisy observation of the corresponding latent component:
\begin{equation}
  \tilde{s}_k \mid z_k
    \;\sim\; \mathcal{N}\!\left(z_k,\; \sigma_k^2\right),
    \quad k \in [K],
  \label{eq:measurement}
\end{equation}
where $\sigma_k^2 > 0$ is the measurement-noise variance of concept $k$. Stacking the $K$ concepts yields $\tilde{\mathbf{s}} \mid \mathbf{z} \sim \mathcal{N}(\mathbf{z}, \mathbf{D})$ with $\mathbf{D} = \diag(\sigma_1^2, \ldots, \sigma_K^2)$.

Under this measurement-error model, the posterior $\mathbf{z} \mid \tilde{\mathbf{s}}$ remains Gaussian with
\begin{align}
  \boldsymbol{\Sigma}_{\mathrm{post}} &= \left(\boldsymbol{\Omega} + \mathbf{D}^{-1}\right)^{-1},
  \label{eq:Sigma_post} \\[2pt]
  \boldsymbol{\mu}_{\mathrm{post}}  &= \boldsymbol{\Sigma}_{\mathrm{post}}\, \mathbf{D}^{-1}\, \tilde{\mathbf{s}}
            \;=\; \mathbf{A}\,\tilde{\mathbf{s}},
  \label{eq:mu_post}
\end{align}
where $\mathbf{A} \triangleq \boldsymbol{\Sigma}_{\mathrm{post}}\,\mathbf{D}^{-1}$ is the denoising matrix. The posterior mean $\boldsymbol{\mu}_{\mathrm{post}}$ is the corrected concept score used by the task head. Intuitively, concepts with larger estimated noise are shrunk more strongly, while off-diagonal structure in $\boldsymbol{\Omega}$ allows related rubric dimensions to share statistical strength.

\begin{proposition}
  \label{prop:mmse}
  Under the Gaussian measurement model in
  Eqs.~\eqref{eq:prior}--\eqref{eq:measurement}, the posterior mean
  $\boldsymbol{\mu}_{\mathrm{post}} = \mathbf{A}\tilde{\mathbf{s}}$ in Eq.~\eqref{eq:mu_post} is the
  unique minimum mean-squared error \textup{(MMSE)} estimator of
  $\mathbf{z}$ given $\tilde{\mathbf{s}}$:
  \[
    \mathbf{A}\tilde{\mathbf{s}}
    = \operatorname*{arg\,min}_{g}
      \;\mathbb{E}\!\left[\|\mathbf{z} - g(\tilde{\mathbf{s}})\|^2\right],
  \]
  with the minimum over all measurable functions $g$.
\end{proposition}
The proof is deferred to Appendix~\ref{app:proof}.

To enforce a positive-definite $\boldsymbol{\Omega}$, we parameterize it through a lower-triangular Cholesky factor $\mathbf{L} \in \mathbb{R}^{K \times K}$:
\begin{equation}
  \boldsymbol{\Omega} = \mathbf{L} \mathbf{L}^\top + \varepsilon \mathbf{I},
  \label{eq:cholesky}
\end{equation}
where $\varepsilon > 0$ is a small regularization constant. Each noise variance is parameterized as $\sigma_k^2 = \exp(\eta_k)$ with $\eta_k \in \mathbb{R}$ learnable, guaranteeing strict positivity.

The corrected concept vector is then passed to a linear task classifier $\mathbf{W} \in \mathbb{R}^{(S+1) \times K}$, producing logits over the $S+1$ grade levels:
\begin{equation}
  \hat{y} = \mathbf{W}\,\boldsymbol{\mu}_{\mathrm{post}}.
  \label{eq:task_head}
\end{equation}
This design preserves the interpretable bottleneck while correcting concept-level noise before grade prediction.

The latent head is trained with three objectives. The task loss supervises the final grade prediction:
\begin{equation}
  \mathcal{L}_{\mathrm{task}} = \mathrm{CrossEntropy}(\hat{y}, y).
  \label{eq:loss_task}
\end{equation}
The denoising alignment loss keeps the corrected concepts close to the normalized rubric annotations $\tilde{\mathbf{c}} = \mathbf{c}/M \in [0,1]^K$:
\begin{equation}
  \mathcal{L}_{\mathrm{den}} = \frac{1}{K}
    \left\|\boldsymbol{\mu}_{\mathrm{post}} - \tilde{\mathbf{c}}\right\|^2.
  \label{eq:loss_den}
\end{equation}
The sparsity penalty regularizes cross-concept dependencies by shrinking the off-diagonal entries of the Cholesky factor:
\begin{equation}
  \mathcal{L}_{\mathrm{spa}} = \sum_{i > j} |L_{ij}|.
  \label{eq:loss_spa}
\end{equation}
This encourages the correction layer to stay near-diagonal unless the data support stronger latent interactions among rubric dimensions.

\subsection{Learning Objective and Training Paradigm.}
\label{sec:training}

The full \MethodName{} training objective combines losses from all components:
\begin{equation}
  \mathcal{L} = \lambda_c \mathcal{L}_{\mathrm{con}} + \lambda_r\mathcal{L}_{\mathrm{rank}}
       + \lambda_t \mathcal{L}_{\mathrm{task}} + \lambda_d \mathcal{L}_{\mathrm{den}} + \lambda_s \mathcal{L}_{\mathrm{spa}},
  \label{eq:total_loss}
\end{equation}
with non-negative hyperparameters $\lambda_c, \lambda_r, \lambda_t, \lambda_d, \lambda_s$.

Optimizing all parameters jointly may entangle the latent head's denoising role with the encoder's representation learning and collapse the interpretable bottleneck. We therefore use two-stage training.

\textbf{Stage~I} jointly trains the text encoder $\mathcal{E}$, concept query bank $\mathcal{Q}$, and per-concept classifiers $\{\mathbf{V}_k\}$ under concept supervision and ordinal calibration objectives:
\begin{equation}
  \min_{\mathcal{E},\;\mathcal{Q},\;\{\mathbf{V}_k\}}
    \;\lambda_c \mathcal{L}_{\mathrm{con}} + \lambda_r\mathcal{L}_{\mathrm{rank}}.
  \label{eq:stage1}
\end{equation}
It thus learns rubric-aligned concept predictors that respect token-level evidence and ordinal structure, after which these components are frozen.

\textbf{Stage~II} fits the latent-correction parameters $(\mathbf{L}, \{\eta_k\}, \mathbf{W})$ on top of the fixed concept outputs:
\begin{equation}
  \min_{\mathbf{L},\;\{\eta_k\},\;\mathbf{W}}
    \;\lambda_t \mathcal{L}_{\mathrm{task}} + \lambda_d \mathcal{L}_{\mathrm{den}} + \lambda_s \mathcal{L}_{\mathrm{spa}}.
  \label{eq:stage2}
\end{equation}
This decomposition ensures that Stage~I yields interpretable, ordinally calibrated concept estimates, while Stage~II learns to correct measurement error and map the corrected concepts to the final grade without distorting the bottleneck. The complete optimization procedure is summarized in Algorithm~\ref{alg:training} in Appendix~\ref{app:algorithm}.

\section{Experiments.}
\label{sec:experiments}
We first introduce the experiment setup and then present comprehensive evaluations.
\subsection{Datasets, Annotation, and Evaluation Protocol.}
\label{sec:datasets}
\emph{Datasets.} We evaluate \MethodName{} on three open-ended grading benchmarks spanning short-answer and essay scoring settings: Mohler short-answer grading~\cite{mohler2011learning}, ASAP 2.0 essay scoring~\cite{crossley2025large}, and MOCHA reading-comprehension answer grading~\cite{chen2020mocha}.

\begin{table}[th]
\centering
\caption{Dataset statistics.}
\vspace{1mm}
\label{tab:dataset_stats}
\resizebox{\linewidth}{!}{
\begin{tabular}{lcccc}
\toprule
\textbf{Dataset} & \textbf{Samples} & \textbf{Concepts} & \textbf{Concept Levels} & \textbf{Grade Levels} \\
\midrule
Mohler & 2,273 & 8 & 3 & 6 \\
\rowcolor{RowGray}
ASAP 2.0 & 17,292 & 8 & 5 & 6 \\
MOCHA & 31,069 & 7 & 3 & 5 \\
\bottomrule
\end{tabular}
}
\end{table}

\begin{table*}[!th]
\scriptsize
\centering
\caption{Main results on the grading benchmarks. \textbf{Bold} marks the best result within each model block. Concept metrics are unavailable for black-box baselines, and ``--'' denotes unavailable values.}
\vspace{1mm}
\label{tab:main_results}
\resizebox{\linewidth}{!}{
\begin{tabular}{ll|cccc|cccc|cccc}
\toprule
\textbf{Dataset} &
\textbf{} &
\multicolumn{4}{c|}{\textbf{Mohler}} &
\multicolumn{4}{c|}{\textbf{ASAP 2.0}} &
\multicolumn{4}{c}{\textbf{MOCHA}} \\
\cmidrule(lr){3-6} \cmidrule(lr){7-10} \cmidrule(lr){11-14}
\textbf{Architecture} &
\textbf{Method} &
\textbf{C-Acc.~$\uparrow$} & \textbf{C-F1~$\uparrow$} & \textbf{T-Acc.~$\uparrow$} & \textbf{T-F1~$\uparrow$} &
\textbf{C-Acc.~$\uparrow$} & \textbf{C-F1~$\uparrow$} & \textbf{T-Acc.~$\uparrow$} & \textbf{T-F1~$\uparrow$} &
\textbf{C-Acc.~$\uparrow$} & \textbf{C-F1~$\uparrow$} & \textbf{T-Acc.~$\uparrow$} & \textbf{T-F1~$\uparrow$} \\
\midrule
\multicolumn{14}{c}{\textit{Pre-trained Language Models (Black-box)}} \\
\midrule
BERT & PLM & -- & -- & 0.706 & 0.529 & -- & -- & 0.591 & 0.485 & -- & -- & 0.578 & 0.463 \\
\rowcolor{RowGray}
BART & PLM & -- & -- & 0.675 & 0.507 & -- & -- & 0.583 & 0.459 & -- & -- & 0.601 & 0.472 \\
GPT-2 & PLM & -- & -- & 0.627 & 0.416 & -- & -- & 0.612 & 0.488 & -- & -- & 0.564 & 0.419 \\
\rowcolor{RowGray}
RoBERTa & PLM & -- & -- & 0.715 & 0.563 & -- & -- & \textbf{0.620} & \textbf{0.526} & -- & -- & \textbf{0.624} & \textbf{0.514} \\
T5-Base & PLM & -- & -- & \textbf{0.741} & \textbf{0.564} & -- & -- & 0.590 & 0.508 & -- & -- & 0.598 & 0.478 \\
\midrule
\multicolumn{14}{c}{\textit{Large Language Models (Black-box)}} \\
\midrule
Llama-3-8B & Zero Shot & -- & -- & 0.158 & 0.207 & -- & -- & 0.334 & 0.175 & -- & -- & 0.305 & 0.219 \\
\rowcolor{RowGray}
Qwen2.5-14B & Zero Shot & -- & -- & 0.070 & 0.104 & -- & -- & \textbf{0.413} & 0.187 & -- & -- & 0.370 & 0.273 \\
Mistral-7B & Zero Shot & -- & -- & 0.329 & \textbf{0.376} & -- & -- & 0.383 & \textbf{0.225} & -- & -- & 0.388 & 0.354 \\
\rowcolor{RowGray}
Mistral-7B & 3-Shot & -- & -- & \textbf{0.399} & 0.352 & -- & -- & 0.277 & 0.221 & -- & -- & \textbf{0.561} & \textbf{0.370} \\
\midrule
\multicolumn{14}{c}{\textit{Concept Bottleneck Models (White-box)}} \\
\midrule
\multirow{5}{*}{\textbf{BERT}} & Vanilla CBM & 0.692 & 0.622 & 0.706 & 0.510 & 0.555 & 0.532 & 0.589 & 0.458 & \textbf{0.734} & 0.565 & 0.586 & 0.473 \\
\rowcolor{RowGray}
& C$^3$M & 0.708 & 0.571 & 0.689 & 0.466 & 0.566 & 0.522 & 0.606 & 0.445 & 0.723 & 0.575 & 0.593 & 0.473 \\
& CT-CBM & 0.638 & 0.522 & 0.583 & 0.237 & 0.492 & 0.478 & 0.598 & 0.449 & 0.669 & 0.525 & 0.576 & 0.446 \\
\rowcolor{RowGray}
& CB-LLM & -- & -- & 0.298 & 0.077 & -- & -- & 0.276 & 0.094 & -- & -- & 0.470 & 0.131 \\
& \MethodName{} & \textbf{0.709} & \textbf{0.648} & \textbf{0.715} & \textbf{0.580} & \textbf{0.573} & \textbf{0.563} & \textbf{0.609} & \textbf{0.466} & 0.731 & \textbf{0.591} & \textbf{0.604} & \textbf{0.486} \\
\addlinespace[2pt]
\cmidrule(lr){1-14}
\multirow{5}{*}{\textbf{BART}} & Vanilla CBM & 0.678 & 0.606 & 0.671 & 0.494 & 0.570 & 0.534 & 0.592 & 0.492 & 0.735 & 0.560 & 0.615 & 0.473 \\
\rowcolor{RowGray}
& C$^3$M & 0.628 & 0.568 & 0.592 & 0.421 & 0.566 & 0.515 & 0.584 & 0.482 & 0.732 & 0.558 & 0.597 & 0.479 \\
& CT-CBM & 0.557 & 0.495 & 0.557 & 0.203 & \textbf{0.601} & \textbf{0.573} & 0.609 & 0.467 & 0.688 & 0.504 & 0.596 & 0.472 \\
\rowcolor{RowGray}
& CB-LLM & -- & -- & 0.298 & 0.077 & -- & -- & 0.276 & 0.094 & -- & -- & 0.470 & 0.131 \\
& \MethodName{} & \textbf{0.707} & \textbf{0.624} & \textbf{0.732} & \textbf{0.541} & 0.578 & 0.513 & \textbf{0.611} & \textbf{0.520} & \textbf{0.738} & \textbf{0.576} & \textbf{0.632} & \textbf{0.497} \\
\addlinespace[2pt]
\cmidrule(lr){1-14}
\multirow{5}{*}{\textbf{GPT-2}} & Vanilla CBM & 0.667 & 0.600 & 0.596 & 0.298 & 0.591 & 0.535 & 0.584 & 0.466 & 0.705 & 0.505 & 0.578 & 0.437 \\
\rowcolor{RowGray}
& C$^3$M & \textbf{0.677} & 0.603 & 0.640 & 0.344 & 0.581 & 0.523 & 0.584 & 0.464 & 0.703 & 0.515 & 0.581 & 0.429 \\
& CT-CBM & 0.488 & 0.429 & 0.496 & 0.110 & 0.346 & 0.280 & 0.363 & 0.089 & 0.399 & 0.278 & 0.475 & 0.129 \\
\rowcolor{RowGray}
& CB-LLM & -- & -- & 0.298 & 0.077 & -- & -- & 0.276 & 0.094 & -- & -- & 0.470 & 0.131 \\
& \MethodName{} & \textbf{0.677} & \textbf{0.623} & \textbf{0.724} & \textbf{0.615} & \textbf{0.595} & \textbf{0.557} & \textbf{0.618} & \textbf{0.507} & \textbf{0.732} & \textbf{0.584} & \textbf{0.617} & \textbf{0.486} \\
\addlinespace[2pt]
\cmidrule(lr){1-14}
\multirow{5}{*}{\textbf{RoBERTa}} & Vanilla CBM & 0.715 & 0.651 & 0.737 & 0.558 & 0.583 & 0.561 & 0.620 & 0.526 & \textbf{0.763} & 0.603 & 0.633 & 0.513 \\
\rowcolor{RowGray}
& C$^3$M & 0.700 & 0.633 & 0.706 & 0.530 & 0.564 & 0.546 & 0.606 & 0.505 & 0.759 & 0.597 & 0.634 & 0.512 \\
& CT-CBM & 0.545 & 0.482 & 0.579 & 0.277 & 0.619 & 0.576 & 0.622 & 0.511 & 0.676 & 0.518 & 0.603 & 0.475 \\
\rowcolor{RowGray}
& CB-LLM & -- & -- & 0.298 & 0.077 & -- & -- & 0.276 & 0.094 & -- & -- & 0.470 & 0.131 \\
& \MethodName{} & \textbf{0.735} & \textbf{0.686} & \textbf{0.746} & \textbf{0.570} & \textbf{0.638} & \textbf{0.607} & \textbf{0.642} & \textbf{0.545} & \textbf{0.763} & \textbf{0.617} & \textbf{0.642} & \textbf{0.530} \\
\addlinespace[2pt]
\cmidrule(lr){1-14}
\multirow{5}{*}{\textbf{T5-Base}} & Vanilla CBM & \textbf{0.702} & 0.622 & 0.711 & 0.484 & 0.593 & 0.555 & 0.571 & 0.473 & 0.729 & 0.558 & 0.608 & 0.481 \\
\rowcolor{RowGray}
& C$^3$M & 0.688 & 0.600 & 0.702 & 0.532 & 0.603 & 0.555 & 0.584 & 0.443 & 0.728 & 0.556 & 0.600 & 0.477 \\
& CT-CBM & 0.589 & 0.518 & 0.610 & 0.267 & 0.610 & 0.576 & 0.595 & 0.463 & 0.734 & 0.570 & 0.618 & 0.483 \\
\rowcolor{RowGray}
& CB-LLM & -- & -- & 0.298 & 0.077 & -- & -- & 0.276 & 0.094 & -- & -- & 0.470 & 0.131 \\
& \MethodName{} & 0.695 & \textbf{0.623} & \textbf{0.724} & \textbf{0.547} & \textbf{0.656} & \textbf{0.616} & \textbf{0.602} & \textbf{0.486} & \textbf{0.774} & \textbf{0.647} & \textbf{0.632} & \textbf{0.500} \\
\bottomrule
\end{tabular}
}
\end{table*}

\emph{Annotation.} None of the three benchmarks include concept-level annotations, so we curate and annotate each dataset with a human-in-the-loop (HITL) annotation pipeline: three domain experts first propose the concepts and corresponding rubrics. Then we prompt GPT-4o~\cite{hurst2024gpt} and Gemini-2.5-pro~\cite{google2025gemini25} to annotate the dataset with proposed rubrics, after which the experts independently review the annotations, identify and resolve inconsistent cases through discussion and majority voting. Dataset-specific descriptions, rubric, and annotation details are deferred to Appendix~\ref{app:datasets_annotation}, and the rubric design rationale, concept groupings, and the rubrics' role in the pipeline are detailed in \S\ref{app:rubric_definitions}. Table~\ref{tab:dataset_stats} summarizes the corpus scale and label-space statistics.

\emph{Evaluation Protocol.} All datasets use a $7{:}2{:}1$ train/dev/test partition. We report task-level accuracy and macro-F1 for grading quality, together with concept-level accuracy and macro-F1 for the intermediate concept predictions. Implementation details such as backbone-specific settings, model selection by validation macro-F1, and stage-specific hyperparameters are deferred to Appendix~\ref{app:implementation_details}. Without specific mention, all results are averaged over five random seeds.

\subsection{Main Results.}
\label{sec:main_results}
\emph{RQ1: Can \MethodName{} improve open-ended grading performance while preserving interpretable concept predictions relative to black-box graders and prior CBM baselines?}

We benchmark \MethodName{} against three baseline classes: (I) fine-tuned PLMs, (II) zero-/few-shot prompted LLMs, and (III) SOTA CBMs adapted to grading; model-specific details appear in Appendix~\ref{app:implementation_details}. Table~\ref{tab:main_results} reports the results. Across three datasets and five backbones, \MethodName{} attains the best T-Acc and T-F1 within every backbone block, and the best C-F1 in 14 of 15 scenarios, improving the bottleneck alongside the final grade. The three baseline classes diverge sharply. Among PLMs, RoBERTa is the strongest black-box grader, yet same-backbone \MethodName{} matches or exceeds its T-F1 on Mohler/ASAP 2.0/MOCHA (0.570/0.545/0.530 vs.\ 0.563/0.526/0.514) while exposing rubric-level evidence that the PLM cannot. Prompted LLMs lag both PLMs and CBMs---T-Acc on Mohler falls to 0.070 (Qwen2.5-14B)---showing that generic instruction-following does not recover rubric-specific grading without task training. Among prior CBMs, Vanilla CBM and C$^3$M are the strongest priors and \MethodName{} improves on both axes across all datasets, whereas CT-CBM and CB-LLM fail on at least one dataset (CB-LLM's T-F1 is constant at 0.077/0.094/0.131 across backbones, indicating majority-class collapse). Critically, accuracy and interpretability move together: on RoBERTa, \MethodName{} lifts C-F1 by $+3.5$/$+4.6$/$+1.4$ points over Vanilla CBM on the three datasets while delivering matching task-side gains. Per-backbone gains and failure-mode diagnostics are in Appendix~\ref{app:main_results_extended}. Overall, \MethodName{} offers a promising white-box grading approach with superior performance and rubric-aligned concept interpretability.

\subsection{Ablation Studies.}
\label{sec:ablation}
\emph{RQ2: Does each component of \MethodName{} contribute a measurable gain to grading accuracy and concept reliability?}

\begin{figure}[!th]
\centering
\includegraphics[width=\linewidth]{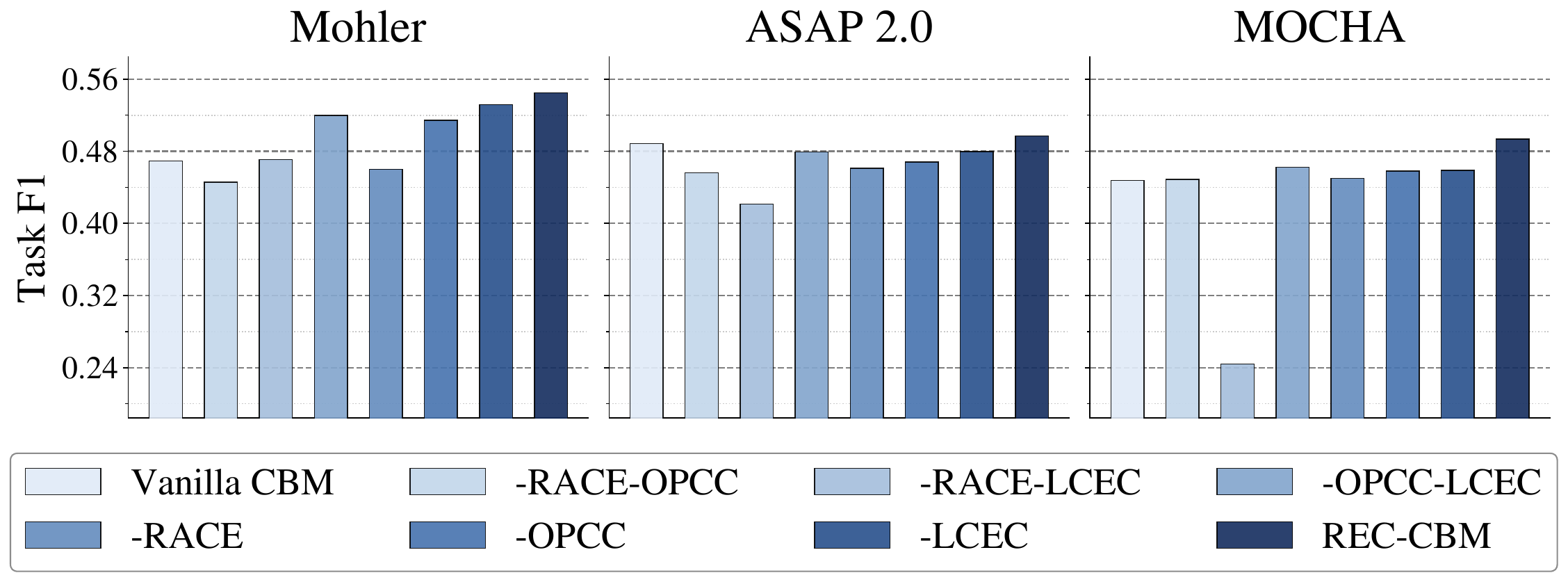}
\vspace{-6mm}
\caption{Ablation over components of \MethodName{}: RACE (rubric-aware encoder), OPCC (ordinal pairwise calibration), and LCEC (latent error correction).}
\label{fig:ablation}
\end{figure}

Fig.~\ref{fig:ablation} compares the full \MethodName{} against seven ablated variants that remove some components: RACE (\S\ref{sec:race}), OPCC (\S\ref{sec:ocr}), and LCEC (\S\ref{sec:lts}). \MethodName{} attains the best T-ACC on every dataset, and removing any single component degrades T-ACC on at least two of the three datasets, indicating that the three modules contribute complementary rather than redundant signal. Among single removals, dropping RACE is the most damaging, consistent with the claim in \S\ref{sec:race} that per-rubric concept queries carry most of the discriminative evidence; removing OPCC or LCEC yields smaller but consistent declines. Double-removal variants sit between the single-removal variants and Vanilla CBM on every dataset, producing a monotone ``more components, better grading'' trend. The sharpest failure is $-$RACE$-$OPCC on MOCHA, where jointly ablating the rubric-aware encoder and the ordinal calibrator nearly halves T-ACC, leaving LCEC to correct a near-uninformative signal on MOCHA's coarse three-level rubric; this identifies RACE together with OPCC as the minimally sufficient upstream for LCEC to be useful. Per-dataset bars and component-wise gain decomposition are reported in Appendix~\ref{app:ablation_extended}.

\subsection{Parameter Analysis.}
\label{sec:param_analysis}
\emph{RQ3: How sensitive is \MethodName{} to its core hyperparameters?}

\begin{figure}[!th]
\centering
\includegraphics[width=\linewidth]{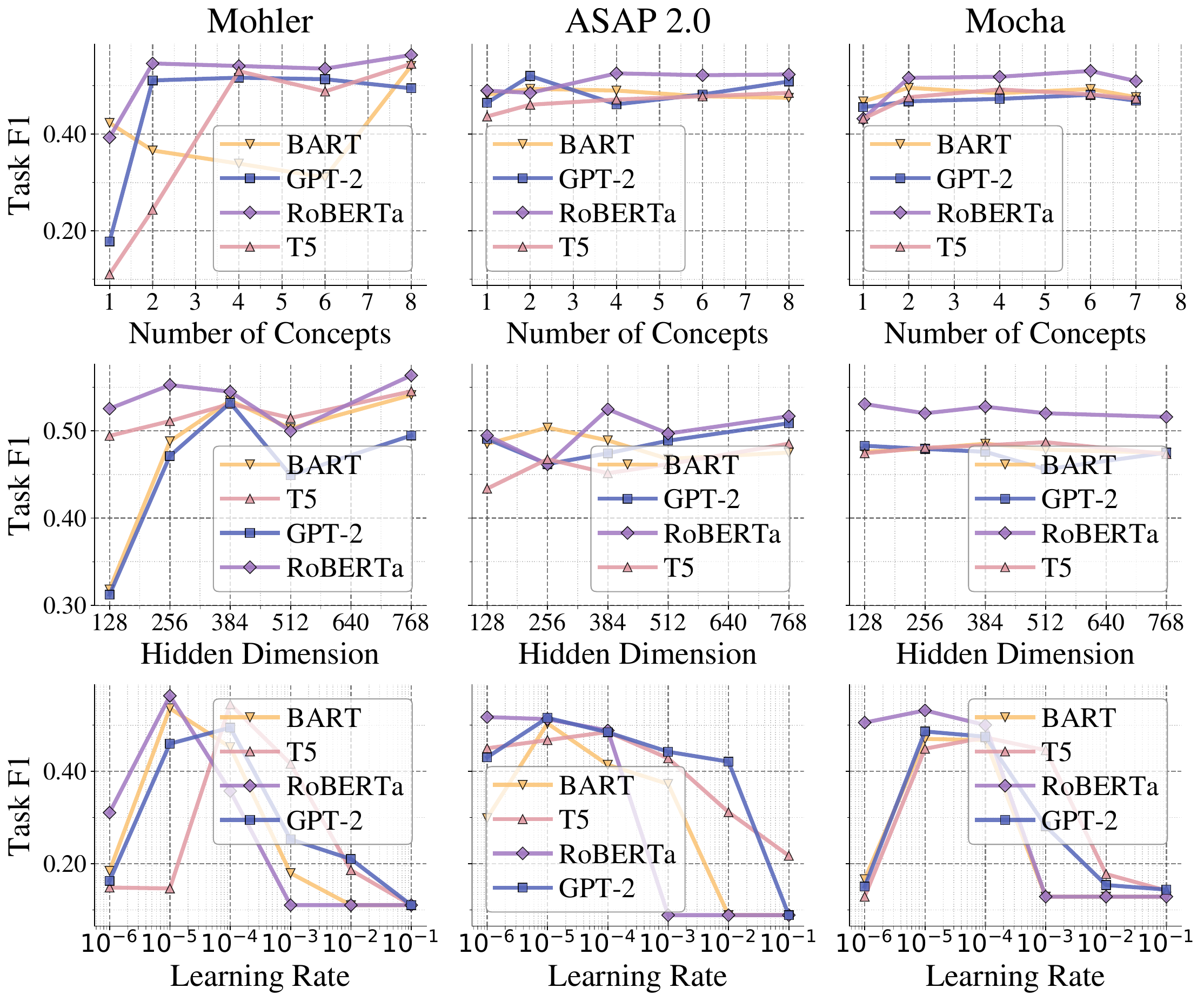}
\vspace{-6mm}
\caption{Parameter analysis of \MethodName{} across various backbones on three grading datasets.}
\label{fig:param_analysis}
\end{figure}

We sweep the three core hyperparameters of \MethodName{} across four backbones on all three datasets and report task F1 in Fig.~\ref{fig:param_analysis}. Two of the three knobs are essentially flat: task F1 saturates by $K{\approx}4$ on every dataset and is insensitive to $d$ once $d \ge 384$, so neither the rubric cardinality nor the encoder width requires backbone-specific tuning. The learning rate is the only stability-critical knob, forming a sharp inverted-U across backbones with a peak in $[10^{-5}, 10^{-4}]$ and a uniform collapse to near-zero F1 once LR$\ge 10^{-3}$. This pattern is consistent with the design in \S\ref{sec:race}--\S\ref{sec:training}: rubric-aware concept queries absorb additional capacity gracefully, while Stage~I jointly updates the encoder, queries, and concept classifiers and is therefore the component most exposed to LR misconfiguration. The narrow tuned LR range used in Table~\ref{tab:impl_keyparams} reflects this sensitivity. Detailed per-backbone analysis and dataset-specific failure modes are reported in Appendix~\ref{app:param_analysis_extended}.

\subsection{Human Intervention.}
\label{sec:human_intervention}
\emph{RQ4: Can educators meaningfully steer \MethodName{} by intervening on its rubric-level concept predictions?}

\begin{figure}[!th]
\centering
\includegraphics[width=\linewidth]{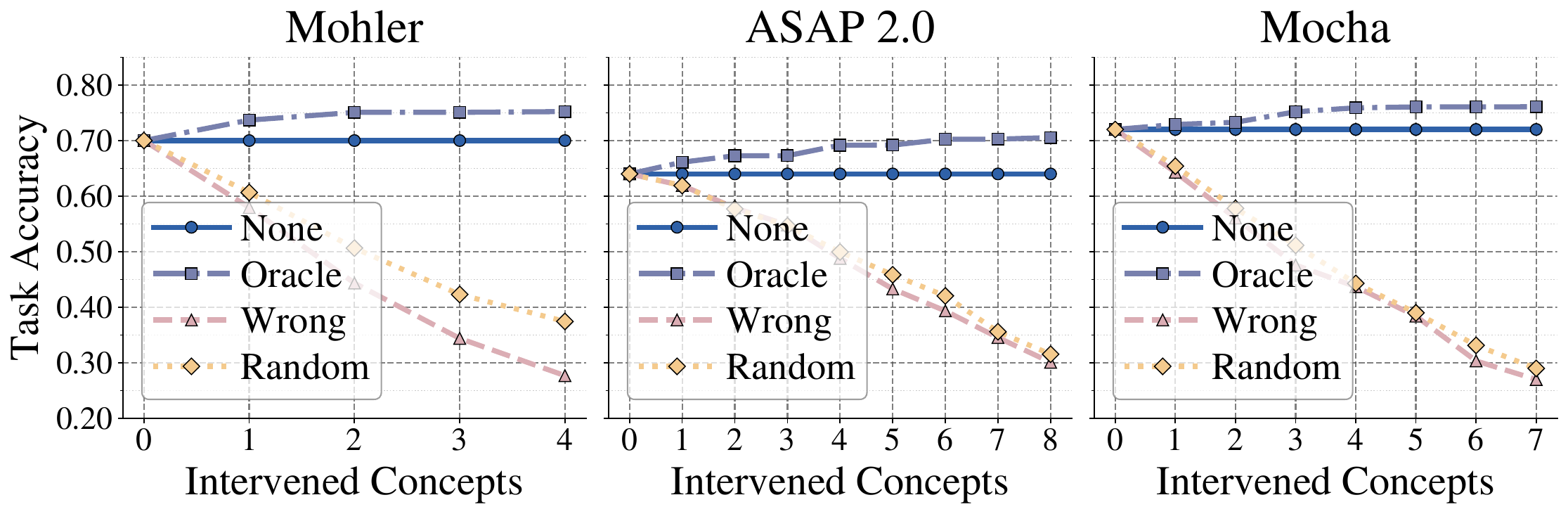}
\vspace{-6mm}
\caption{Human intervention on rubric concepts. For each dataset, we replace the top-$k$ predicted concepts with the oracle rubric label (Oracle), an adversarial wrong label (Wrong), or a uniformly random label (Random), and compare against no intervention (None).}
\label{fig:intervention}
\end{figure}

To test whether the bottleneck is actionable, we simulate educator interventions on the top-$k$ highest-confidence predicted concepts and re-evaluate the frozen Stage~II head (Fig.~\ref{fig:intervention}). Wrong and Random degrade monotonically, losing 30--45 accuracy points at the largest $k$ on every dataset, while Oracle matches or slightly exceeds None. This asymmetry is the key faithfulness signal: corrupting rubric evidence flips the predicted grade, whereas correcting it preserves or improves the grade. The small Oracle--None gap reflects the measurement-error framing of \S\ref{sec:lts}---the Stage~II head consumes Stage~I concept outputs, so one-hot oracle substitutions are mildly off-distribution. Practically, educators can audit or override individual rubric dimensions and observe the grade respond accordingly, realizing the intervention property that motivates \MethodName{}. Per-dataset curves and Wrong--Random differentials are deferred to Appendix~\ref{app:intervention_extended}.

\subsection{Latent Denoising Analysis.}
\label{sec:denoising_analysis}
\emph{RQ5: Does the latent correction module learn a sparse, rubric-meaningful concept dependency structure rather than echoing raw label correlations?}

\begin{figure}[!th]
\centering
\includegraphics[width=\linewidth]{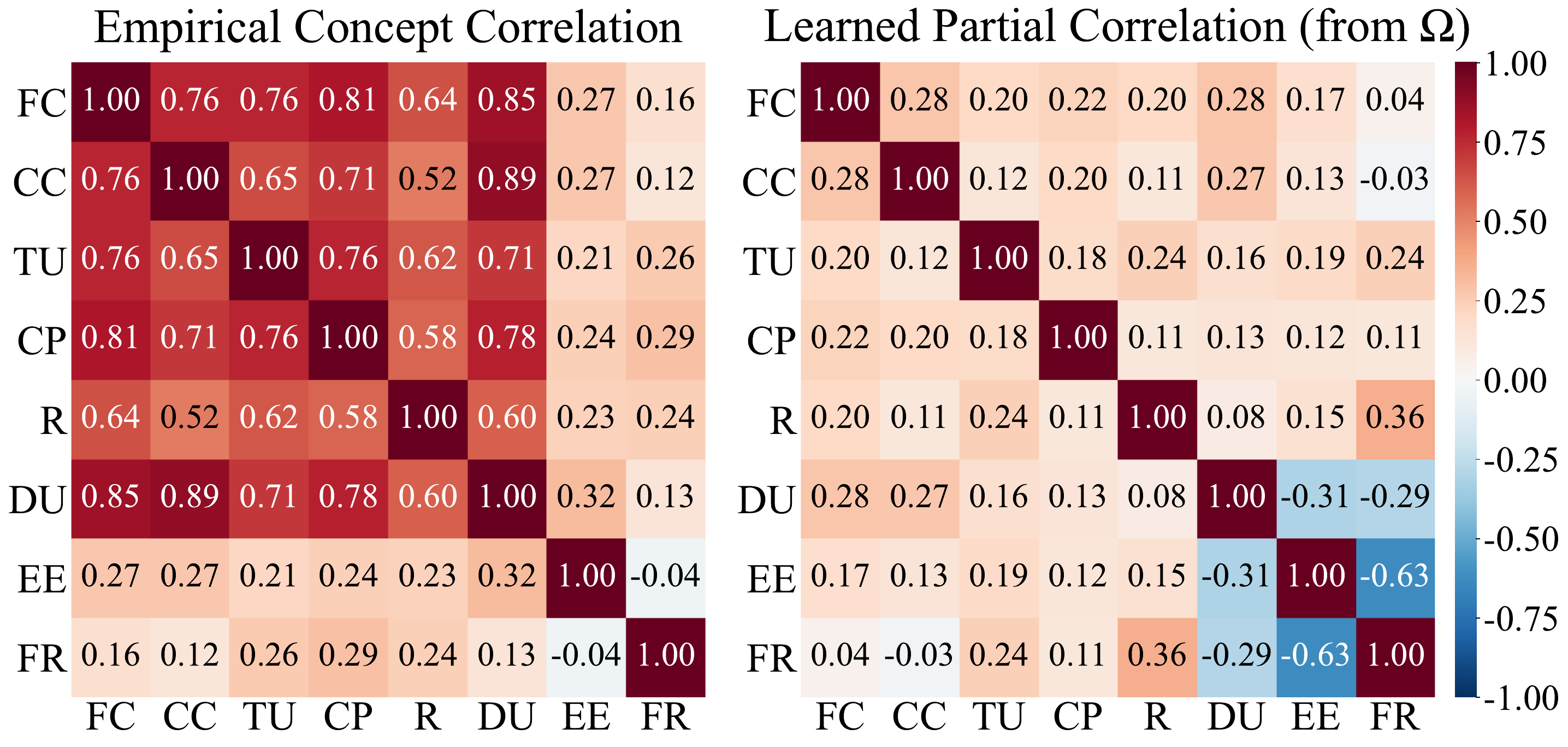}
\vspace{-6mm}
\caption{Latent denoising analysis: empirical rubric-label correlation (left) versus learned partial correlation from the latent precision $\boldsymbol{\Omega}$ (right).}
\label{fig:denoising_analysis}
\end{figure}

Fig.~\ref{fig:denoising_analysis} contrasts the empirical rubric-label correlation (left) with the partial correlations induced by the learned precision $\boldsymbol{\Omega}$ in Eq.~\eqref{eq:prior} (right). The empirical side shows dense co-correlation across content and reasoning dimensions, confirming the annotator overlap and rubric redundancy that motivate \S\ref{sec:lts}. The learned $\boldsymbol{\Omega}$ is markedly sparser, consistent with the near-diagonal target enforced by $\mathcal{L}_{\mathrm{spa}}$ in Eq.~\eqref{eq:loss_spa}, while perserves most interpretable dependencies. Consequently, the denoising matrix $\mathbf{A}=\boldsymbol{\Sigma}_{\mathrm{post}}\mathbf{D}^{-1}$ propagates information along a sparse rubric-level graph that decouples correlated concepts and amplifies the most discriminative dimensions. Deatailed heatmap analysis and the named residual links are deferred to Appendix~\ref{app:denoising_extended}.

\begin{figure}[h]
\center
\includegraphics[width=\columnwidth]{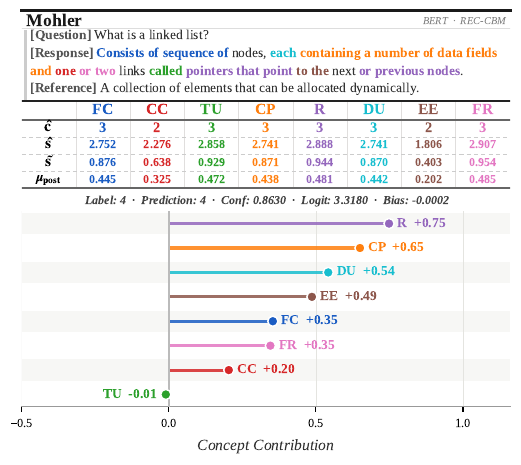}
\vspace{-8mm}
\caption{End-to-end decision trace for a Mohler short-answer response under BERT-\MethodName{}. From top to bottom: rubric-aware token attention, per-concept ordinal predictions ($\hat{\mathbf{c}}$, $\hat{\mathbf{s}}$) with their LCEC input ($\tilde{\mathbf{s}}$) and posterior ($\boldsymbol{\mu}_{\mathrm{post}}$), and each concept's contribution to the predicted grade.}
\label{fig:case_study_qa}
\vspace{-4mm}
\end{figure}

\subsection{Case Study.}
\label{sec:case_study}
\emph{RQ6: Does \MethodName{} produce a human-readable decision trace that an educator can audit end-to-end on a single response?}

Fig.~\ref{fig:case_study_qa} traces a single Mohler short-answer grading decision by a BERT-backbone \MethodName{}, whose predicted grade matches the human label with high confidence. The panel exposes three strata of the bottleneck: rubric-aware token attention from \S\ref{sec:race} (colored response spans), per-concept ordinal scores calibrated by \S\ref{sec:ocr} together with their LCEC-corrected posteriors from \S\ref{sec:lts} (table rows), and the additive per-concept contributions that compose the final grade logit (bar chart). Every quantity in the trace is tied to a named rubric dimension, so the grade is reconstructible from the bottleneck and any dimension can be inspected or overridden, realizing the inspect/verify/intervene workflow motivated in \S1. A component-by-component walkthrough of this example is deferred to Appendix~\ref{app:case_study_extended}.

\section{Conclusion.}
In this paper, we present \MethodName{}, a concept bottleneck framework for trustworthy open-ended grading that integrates a rubric-aware concept encoder, an ordinal pairwise concept calibration, and a latent concept error correction module. Extensive experiments and analyses demonstrate that \MethodName{} attains stronger grading performance while delivering more faithful, robust, and actionable rubric-level reasoning than existing black-box and transparent approaches. Future work will extend \MethodName{} to multilingual and domain-specific assessment and couple it with educator-in-the-loop rubric refinement.


\bibliographystyle{aaai}
\bibliography{references}

\newpage
\onecolumn
\appendix
\section{Notation.}
\label{app:notation}

Table~\ref{tab:notation} summarizes the symbols used in \S\ref{sec:method} and Appendix~\ref{app:method_details}, grouped by role: inputs and labels, rubric-aware encoder, ordinal calibration, latent error correction, and training objective.

\begin{table*}[!h]
\centering
\caption{Symbol list for \MethodName{}.}
\label{tab:notation}
\vspace{1mm}
\resizebox{\linewidth}{!}{
\begin{tabular}{@{}llll@{}}
\toprule
\textbf{Symbol} & \textbf{Meaning} & \textbf{Type} & \textbf{Range / Space} \\
\midrule
\multicolumn{4}{c}{\textit{Inputs and labels}} \\
\cmidrule(lr){1-4}
$\mathbf{x}=(q,r,a)$      & Grading instance (question, response, auxiliary context)       & Tuple    & $\mathcal{X}$ \\
\rowcolor{RowGray}
$y,\ \hat{y}$             & Grade label / predicted grade logits   & Scalar / vector & $\mathcal{Y}=\{0,\ldots,S\}$ / $\mathbb{R}^{S+1}$ \\
$K$                       & Number of rubric concepts & Scalar & $\mathbb{N}$ \\
\rowcolor{RowGray}
$M$                       & Maximum ordinal level per concept & Scalar & $\mathbb{N}$ \\
$c_k,\ \mathbf{c}$        & Concept-level label (scalar entry / stacked vector) & Ordinal / vector & $\{0,\ldots,M\}$ / $\{0,\ldots,M\}^K$ \\
\rowcolor{RowGray}
$\tilde{c}_k,\ \tilde{\mathbf{c}}$ & Normalized concept label $\tilde{c}_k=c_k/M$ & Scalar / vector & $[0,1]$ / $[0,1]^K$ \\
\midrule
\multicolumn{4}{c}{\textit{Rubric-aware encoder (\S\ref{sec:race})}} \\
\cmidrule(lr){1-4}
$\mathcal{E}$             & Text encoder (pretrained backbone) & Function & $\mathcal{X}\to\mathbb{R}^{T\times d}$ \\
\rowcolor{RowGray}
$T,\ d$                   & Sequence length / hidden dimension & Scalar & $\mathbb{N}$ \\
$\mathbf{H},\ \mathbf{H}_t$ & Token representations / $t$-th token state & Matrix / vector & $\mathbb{R}^{T\times d}$ / $\mathbb{R}^d$ \\
\rowcolor{RowGray}
$\mathcal{Q}=\{\mathbf{q}_k\}$ & Concept query bank  & Set of vectors & $\mathbb{R}^d$ \\
$\alpha_{k,t}$            & Soft attention weight of concept $k$ on token $t$  & Scalar   & $[0,1]$ \\
\rowcolor{RowGray}
$\mathbf{h}_k$            & Concept-$k$ aggregated representation & Vector & $\mathbb{R}^d$ \\
$\tau$                    & Attention temperature  & Scalar   & $\mathbb{R}_{>0}$ \\
\rowcolor{RowGray}
$\mathbf{V}_k$            & Per-concept classifier & Matrix   & $\mathbb{R}^{(M+1)\times d}$ \\
$\hat{\mathbf{p}}_k$      & Predicted ordinal distribution for concept $k$ & Distribution & $\mathbb{R}^{M+1}$ \\
\rowcolor{RowGray}
$\hat{p}_{k,m}$           & $m$-th entry of $\hat{\mathbf{p}}_k$ & Scalar & $[0,1]$ \\
\midrule
\multicolumn{4}{c}{\textit{Ordinal calibration (\S\ref{sec:ocr})}} \\
\cmidrule(lr){1-4}
$\hat{c}_k,\ \hat{\mathbf{c}}$ & Predicted (expected) concept score & Scalar / vector & $[0,M]$ / $[0,M]^K$ \\
\rowcolor{RowGray}
$B$                       & Mini-batch size        & Scalar   & $\mathbb{N}$ \\
$\mathcal{P}_k$           & Within-batch valid comparison pairs & Set & $\subseteq [B]\times [B]$ \\
\rowcolor{RowGray}
$K'$                      & \# concepts with non-empty $\mathcal{P}_k$ & Scalar & $\{0,\ldots,K\}$ \\
$\sigma(\cdot)$           & Logistic sigmoid       & Function & $\mathbb{R}\to(0,1)$ \\
\midrule
\multicolumn{4}{c}{\textit{Latent error correction (\S\ref{sec:lts})}} \\
\cmidrule(lr){1-4}
$\tilde{s}_k,\ \tilde{\mathbf{s}}$ & Normalized concept score $\tilde{s}_k=\hat{c}_k/M$ & Scalar / vector & $[0,1]$ / $[0,1]^K$ \\
\rowcolor{RowGray}
$\mathbf{z}$              & Latent concept vector  & Vector   & $\mathbb{R}^K$ \\
$\boldsymbol{\Omega}$     & Latent precision matrix & Matrix (PD) & $\mathbb{R}^{K\times K}$ \\
\rowcolor{RowGray}
$\mathbf{L}$              & Cholesky factor of $\boldsymbol{\Omega}-\varepsilon\mathbf{I}$ & Lower-triangular matrix & $\mathbb{R}^{K\times K}$ \\
$\varepsilon$             & Precision regularizer  & Scalar   & $\mathbb{R}_{>0}$ \\
\rowcolor{RowGray}
$\eta_k$                  & Log-variance parameter & Scalar   & $\mathbb{R}$ \\
$\sigma_k^2=\exp(\eta_k)$ & Measurement-noise variance of concept $k$   & Scalar   & $\mathbb{R}_{>0}$ \\
\rowcolor{RowGray}
$\mathbf{D}$              & Diagonal measurement covariance & Diagonal matrix & $\mathbb{R}^{K\times K}$ \\
$\boldsymbol{\Sigma}_{\mathrm{post}}$ & Posterior covariance of $\mathbf{z}\mid\tilde{\mathbf{s}}$ & Matrix (PD) & $\mathbb{R}^{K\times K}$ \\
\rowcolor{RowGray}
$\boldsymbol{\mu}_{\mathrm{post}}$    & Posterior mean (corrected concept vector) & Vector & $\mathbb{R}^K$ \\
$\mathbf{A}=\boldsymbol{\Sigma}_{\mathrm{post}}\mathbf{D}^{-1}$ & Denoising matrix & Matrix & $\mathbb{R}^{K\times K}$ \\
\rowcolor{RowGray}
$\mathbf{W}$              & Task classifier (grade head) & Matrix   & $\mathbb{R}^{(S+1)\times K}$ \\
\midrule
\multicolumn{4}{c}{\textit{Training objective (\S\ref{sec:training})}} \\
\cmidrule(lr){1-4}
$\mathcal{L}_{\mathrm{con}},\ \mathcal{L}_{\mathrm{rank}}$ & Concept / ordinal ranking losses & Scalar & $\mathbb{R}_{\ge 0}$ \\
\rowcolor{RowGray}
$\mathcal{L}_{\mathrm{task}},\ \mathcal{L}_{\mathrm{den}},\ \mathcal{L}_{\mathrm{spa}}$ & Task / denoising / sparsity losses & Scalar & $\mathbb{R}_{\ge 0}$ \\
$\lambda_c,\lambda_r,\lambda_t,\lambda_d,\lambda_s$ & Loss weights & Scalar & $\mathbb{R}_{\ge 0}$ \\
\rowcolor{RowGray}
$E_{\mathrm{I}},\ E_{\mathrm{II}}$ & Stage-I / Stage-II epoch budgets & Scalar & $\mathbb{N}$ \\
\bottomrule
\end{tabular}
}
\end{table*}

\section{Method Details.}
\label{app:method_details}

This appendix collects the supporting derivation and the full training algorithm referenced in \S\ref{sec:method}.

\subsection{Proof of Proposition~\ref{prop:mmse}.}
\label{app:proof}

We now justify the MMSE claim in Proposition~\ref{prop:mmse}. Recall from \S\ref{sec:lts} that the latent correction module assumes a Gaussian prior over the latent concept vector and a Gaussian measurement model for the normalized concept scores. Under these assumptions, the posterior mean admits a closed form and coincides with the Bayes estimator under squared error.

For reference, we first record the equivalent conditional-mean form
\[
  \boldsymbol{\mu}_{\mathrm{post}}
  = \boldsymbol{\Omega}^{-1}\!\left(\boldsymbol{\Omega}^{-1} + \mathbf{D}\right)^{-1}\!\tilde{\mathbf{s}}
  = \left(\boldsymbol{\Omega} + \mathbf{D}^{-1}\right)^{-1}\!\mathbf{D}^{-1}\tilde{\mathbf{s}},
\]
which makes explicit how the corrected concept vector balances latent precision against measurement noise.

\begin{proof}
  By Eqs.~\eqref{eq:prior}--\eqref{eq:measurement}, we may write the
  observation model in vector form as
  \[
    \tilde{\mathbf{s}} = \mathbf{z} + \boldsymbol{\varepsilon},
    \qquad
    \mathbf{z} \sim \mathcal{N}(\mathbf{0}, \boldsymbol{\Omega}^{-1}),
    \qquad
    \boldsymbol{\varepsilon} \sim \mathcal{N}(\mathbf{0}, \mathbf{D}),
  \]
  where $\boldsymbol{\varepsilon}$ is independent of $\mathbf{z}$. It follows
  that $(\mathbf{z}, \tilde{\mathbf{s}})$ is jointly Gaussian with mean zero and
  block covariance
  \[
    \mathrm{Cov}
    \begin{pmatrix}
      \mathbf{z} \\
      \tilde{\mathbf{s}}
    \end{pmatrix}
    =
    \begin{pmatrix}
      \boldsymbol{\Omega}^{-1} & \boldsymbol{\Omega}^{-1} \\
      \boldsymbol{\Omega}^{-1} & \boldsymbol{\Omega}^{-1} + \mathbf{D}
    \end{pmatrix}.
  \]
  Indeed, $\mathrm{Cov}(\mathbf{z}, \tilde{\mathbf{s}}) =
  \mathrm{Cov}(\mathbf{z}, \mathbf{z} + \boldsymbol{\varepsilon}) =
  \boldsymbol{\Omega}^{-1}$ because the measurement noise is independent of
  $\mathbf{z}$, and $\mathrm{Cov}(\tilde{\mathbf{s}}, \tilde{\mathbf{s}}) =
  \boldsymbol{\Omega}^{-1} + \mathbf{D}$.

  For any square-integrable random vector, the estimator minimizing
  $\mathbb{E}[\|\mathbf{z} - g(\tilde{\mathbf{s}})\|^2]$ over all measurable
  functions $g$ is the conditional expectation
  $\mathbb{E}[\mathbf{z} \mid \tilde{\mathbf{s}}]$. Since the pair
  $(\mathbf{z}, \tilde{\mathbf{s}})$ is jointly Gaussian, this conditional
  expectation is affine, and because the mean is zero it is in fact linear:
  \[
    \mathbb{E}[\mathbf{z} \mid \tilde{\mathbf{s}}]
    =
    \mathrm{Cov}(\mathbf{z}, \tilde{\mathbf{s}})
    \,\mathrm{Cov}(\tilde{\mathbf{s}}, \tilde{\mathbf{s}})^{-1}
    \tilde{\mathbf{s}}
    =
    \boldsymbol{\Omega}^{-1}
    \left(\boldsymbol{\Omega}^{-1} + \mathbf{D}\right)^{-1}
    \tilde{\mathbf{s}}.
  \]

  It remains to express this linear estimator in the posterior form used in
  the main text. Using the identity
  $P(P + R)^{-1} = (I + RP^{-1})^{-1} = (P^{-1} + R^{-1})^{-1}R^{-1}$
  with $P = \boldsymbol{\Omega}^{-1}$ and $R = \mathbf{D}$, we obtain
  \[
    \mathbb{E}[\mathbf{z} \mid \tilde{\mathbf{s}}]
    =
    \left(\boldsymbol{\Omega} + \mathbf{D}^{-1}\right)^{-1}
    \mathbf{D}^{-1}\tilde{\mathbf{s}}
    =
    \boldsymbol{\Sigma}_{\mathrm{post}}\,\mathbf{D}^{-1}\tilde{\mathbf{s}}
    =
    \mathbf{A}\tilde{\mathbf{s}}
    =
    \boldsymbol{\mu}_{\mathrm{post}}.
  \]
  Therefore, $\boldsymbol{\mu}_{\mathrm{post}}$ is the MMSE estimator of
  $\mathbf{z}$ given $\tilde{\mathbf{s}}$. Uniqueness holds up to almost-sure
  equality because conditional expectation is unique in $L^2$.
\end{proof}

\begin{remark}
  In the diagonal case
  $\boldsymbol{\Omega} = \diag(\sigma_{z_1}^{-2}, \ldots, \sigma_{z_K}^{-2})$,
  the denoising matrix becomes diagonal and the $k$-th entry simplifies to
  \[
    A_{kk}
    = \frac{\sigma_{z_k}^2}{\sigma_{z_k}^2 + \sigma_k^2}
    = \rho_k,
  \]
  recovering Spearman's reliability coefficient and the scalar Wiener filter for concept~$k$~\cite{spearman1904proof}.
  In the correlated case, the off-diagonal entries of $\boldsymbol{\Omega}$
  propagate information from reliably measured concepts to noisily
  measured ones, implementing a multivariate generalization of
  reliability-weighted shrinkage that is optimal in the MMSE sense.
\end{remark}

\subsection{Two-Stage Training Algorithm.}
\label{app:algorithm}

Algorithm~\ref{alg:training} expands the two-stage optimization procedure described in \S\ref{sec:training}. The key design choice is to decouple \emph{concept learning} from \emph{error-aware grade prediction}. In Stage~I, the model learns rubric-aligned concept predictors directly from the response text, while the ordinal ranking objective encourages concept estimates to preserve the relative ordering induced by rubric labels. In Stage~II, these concept predictors are frozen and treated as fixed measurements, allowing the latent correction module to focus on denoising and aggregating concept evidence for final grade prediction without altering the interpretable bottleneck.

This stage-wise decomposition serves two purposes. First, it stabilizes optimization by preventing the latent task head from backpropagating through the concept encoder and distorting the rubric-specific representations. Second, it preserves the semantics of the concept layer: the concept scores remain grounded in textual evidence learned in Stage~I, while Stage~II only adjusts how these noisy concept measurements are combined downstream. The resulting procedure therefore maintains reviewer-facing interpretability while improving robustness to annotation noise and cross-concept dependency.

\begin{algorithm}[th]
\caption{Two-stage training procedure for \MethodName{}}
\label{alg:training}
\begin{algorithmic}[1]
\Require Training set $\mathcal{D} = \{(\mathbf{x}^{(n)}, y^{(n)}, \mathbf{c}^{(n)})\}_{n=1}^{N}$
\Require Loss weights $\lambda_c, \lambda_r, \lambda_t, \lambda_d, \lambda_s$ and epoch budgets $E_{\mathrm{I}}, E_{\mathrm{II}}$
\Ensure Trained parameters $\Theta = (\mathcal{E}, \mathcal{Q}, \{\mathbf{V}_k\}, \mathbf{L}, \{\eta_k\}, \mathbf{W})$
\State Initialize encoder $\mathcal{E}$ and concept-query bank $\mathcal{Q} = \{\mathbf{q}_k\}_{k=1}^{K}$
\State Initialize concept classifiers $\{\mathbf{V}_k\}_{k=1}^{K}$, latent parameters $(\mathbf{L}, \{\eta_k\})$, and task head $\mathbf{W}$
\Statex \textbf{Stage I: Learn rubric-aligned ordinal concept predictors}
\For{$e \gets 1$ \textbf{to} $E_{\mathrm{I}}$}
  \ForAll{minibatches $\mathcal{B} \subset \mathcal{D}$}
    \State Encode each response in $\mathcal{B}$ with $\mathcal{E}$ to obtain token states $\mathbf{H}$
    \State Apply concept queries $\mathcal{Q}$ to compute rubric-specific representations $\{\mathbf{h}_k\}_{k=1}^{K}$
    \State Predict concept distributions $\{\hat{\mathbf{p}}_k\}_{k=1}^{K}$ with $\{\mathbf{V}_k\}_{k=1}^{K}$
    \State Compute expected concept scores $\{\hat{c}_k\}_{k=1}^{K}$ from $\{\hat{\mathbf{p}}_k\}_{k=1}^{K}$
    \State Construct within-batch comparison sets $\{\mathcal{P}_k\}_{k=1}^{K}$ from concept labels
    \State Evaluate concept loss $\mathcal{L}_{\mathrm{con}}$ and ranking loss $\mathcal{L}_{\mathrm{rank}}$
    \State Update $(\mathcal{E}, \mathcal{Q}, \{\mathbf{V}_k\})$ using $\lambda_c \mathcal{L}_{\mathrm{con}} + \lambda_r \mathcal{L}_{\mathrm{rank}}$
  \EndFor
\EndFor
\State Freeze $\mathcal{E}$, $\mathcal{Q}$, and $\{\mathbf{V}_k\}$
\Statex \textbf{Stage II: Learn latent correction and grade prediction}
\For{$e \gets 1$ \textbf{to} $E_{\mathrm{II}}$}
  \ForAll{minibatches $\mathcal{B} \subset \mathcal{D}$}
    \State Recompute frozen concept predictions $\{(\hat{\mathbf{p}}_k, \hat{c}_k)\}_{k=1}^{K}$ for $\mathcal{B}$ without gradients
    \State Normalize concept scores: $\tilde{s}_k \leftarrow \hat{c}_k / M$ for all $k \in [K]$
    \State Form $\mathbf{D} = \diag(\exp(\eta_1), \ldots, \exp(\eta_K))$ and $\boldsymbol{\Omega} = \mathbf{L}\mathbf{L}^{\top} + \varepsilon \mathbf{I}$
    \State Compute posterior mean $\boldsymbol{\mu}_{\mathrm{post}}$ via Eq.~\eqref{eq:mu_post}
    \State Predict grade logits $\hat{y} \leftarrow \mathbf{W}\,\boldsymbol{\mu}_{\mathrm{post}}$
    \State Evaluate $\mathcal{L}_{\mathrm{task}}$, $\mathcal{L}_{\mathrm{den}}$, and $\mathcal{L}_{\mathrm{spa}}$
    \State Update $(\mathbf{L}, \{\eta_k\}, \mathbf{W})$ using $\lambda_t \mathcal{L}_{\mathrm{task}} + \lambda_d \mathcal{L}_{\mathrm{den}} + \lambda_s \mathcal{L}_{\mathrm{spa}}$
  \EndFor
\EndFor
\State \Return $\Theta = (\mathcal{E}, \mathcal{Q}, \{\mathbf{V}_k\}, \mathbf{L}, \{\eta_k\}, \mathbf{W})$
\end{algorithmic}
\end{algorithm}

Compared with a one-stage end-to-end update, Algorithm~\ref{alg:training} makes the information flow explicit. Stage~I determines \emph{what} rubric evidence is extracted from text, and Stage~II determines \emph{how} that evidence should be reliability-weighted and aggregated for grading. This separation is especially useful in open-ended assessment, where concept annotations are informative but imperfect, because it lets the latent head correct noisy concept measurements without overwriting the rubric-aligned concept definitions learned from the response text.

\section{Datasets and Annotation.}
\label{app:datasets_annotation}

This appendix details the benchmarks and the rubric-aligned human-in-the-loop annotation pipeline used throughout \S\ref{sec:datasets}.

\subsection{Benchmark Details.}
\label{app:benchmark_details}

We use three public grading benchmarks covering complementary open-ended assessment settings. Mohler~\cite{mohler2011learning} is a short-answer grading benchmark in computer science, where each instance contains a question, a reference answer, and a student response. ASAP 2.0~\cite{crossley2025large} is an essay-scoring benchmark containing longer student compositions graded on an ordinal holistic scale. MOCHA~\cite{chen2020mocha} is a reading-comprehension answer-grading benchmark in which each response is evaluated with respect to a question, supporting context, and reference answer.


\subsection{Annotation Pipeline.}
\label{app:annotation_rubrics}

None of the three benchmarks includes concept-level supervision, so we curate rubric-aligned concept annotations for all datasets. We use a human-in-the-loop workflow involving three domain experts. The experts first propose the concept inventory and corresponding rubrics for each benchmark, specifying the ordinal interpretation and anchored descriptors of every concept level. Given these rubric definitions, GPT-4o~\cite{hurst2024gpt} and Gemini-2.5-pro~\cite{google2025gemini25} produce initial concept annotations for each response.

The same three experts then validate the LLM-generated annotations. They independently review the proposed labels, identify inconsistent or ambiguous cases, and resolve disagreements through discussion and majority voting. This pipeline combines scalable initial labeling with expert oversight, while preserving a transparent decision rule for the final concept labels.

To make this annotation stage reproducible, we use a shared prompt template that conditions on benchmark-specific input fields and rubric definitions, while keeping the output interface fixed across datasets. The template is used only to generate the initial concept annotations; all final labels still come from the expert-validation procedure described above.

\begin{tcolorbox}[promptstyle, title=Prompt for LLM Concept Annotation]
\small
\ttfamily
You are an expert educational assessor for open-ended grading. Your task is to assign one rubric-aligned ordinal label to each concept for a student response, using only the provided benchmark inputs and rubric definitions.\\

Inputs:\\
Benchmark: <benchmark\_name>\\
Question or prompt (q): <question\_or\_prompt>\\
Student response (r): <student\_response>\\
Auxiliary grading context (a): <reference\_answer>, <concept\_list>, <rubric\_definitions>\\

Optional in-context examples:\\
Example 1 input: <example\_1\_question>, <example\_1\_reference>, <example\_1\_rubric>, <example\_1\_response>\\
Example 1 output: <json\_annotation\_list\_1>\\
Example 2 input: <example\_2\_question>, <example\_2\_reference>, <example\_2\_rubric>, <example\_2\_response>\\
Example 2 output: <json\_annotation\_list\_2>\\

Requirements:\\
1. Score each concept independently using only the rubric definitions and the supplied benchmark context.\\
2. Use the reference answer, concept list, and rubric descriptors when evaluating the student response.\\
3. Respect the ordered level meanings defined in the rubrics; do not interpolate between levels or invent new labels.\\
4. For each concept, provide a brief evidence snippet grounded in the student response or the supplied grading context.\\
5. If the evidence is ambiguous between two adjacent levels, choose the lower-supported level.\\
6. Do not provide explanations outside the requested output structure.\\

Return only a structured list with one entry per concept in the format:\\
concept=<name>; label=<ordinal\_level>; evidence=<brief\_evidence>
\end{tcolorbox}

The placeholders are instantiated with the benchmark-specific inputs used throughout the experiments. For Mohler, the prompt receives the question, student response, and auxiliary grading context consisting of the reference answer, the eight-concept inventory, and the corresponding three-level rubrics. For ASAP 2.0, it receives the essay prompt, student essay, and auxiliary grading context consisting of the concept inventory and the corresponding five-level rubrics. For MOCHA, it receives the reading context, question, candidate answer, and auxiliary grading context consisting of the reference answer, the seven-concept inventory, and the corresponding three-level rubrics. Grade labels follow the native benchmark scales summarized in Table~\ref{tab:dataset_stats}.

\subsection{Rubric Definitions.}
\label{app:rubric_definitions}

Tables~\ref{tab:rubric_asap}--\ref{tab:rubric_mocha} list the rubric dimensions for ASAP~2.0, Mohler, and MOCHA, including the textual definition, an illustrative excerpt with the assigned score, and the ordinal score scale. These rubrics are the inputs to the annotation pipeline in \S\ref{app:annotation_rubrics} and the supervision targets for the concept encoder in \S\ref{sec:race}. ASAP~2.0 uses eight concepts on a five-level scale; Mohler uses eight concepts on a three-level scale; MOCHA uses seven concepts on a three-level scale.

\emph{Construction principles.} The three domain experts described in \S\ref{app:annotation_rubrics} drafted each inventory under three criteria. First, the chosen concepts should \emph{cover} the dominant rubric axes that determine the holistic grade for the benchmark, so that grading evidence is not pushed outside the bottleneck. Second, concepts should be \emph{non-redundant}, with each dimension contributing distinct evidence rather than re-scoring an aspect already captured by a sibling concept. Third, every level of every concept should be \emph{ordinally interpretable}, with an anchored descriptor that an expert can apply consistently. These criteria match how the rubrics are subsequently consumed: the concept encoder in \S\ref{sec:race} treats each rubric dimension as a separate supervision target, so coverage and non-redundancy directly determine whether the bottleneck is sufficient for the task.

\emph{Why the scale cardinalities differ.} ASAP~2.0 uses a five-level ordinal scale per concept, while Mohler and MOCHA use three. The cardinality reflects the granularity that human raters can reliably assign given the response artifact. Multi-paragraph essays in ASAP~2.0 (\S\ref{app:benchmark_details}) afford enough surface evidence per dimension to distinguish five anchored bands, whereas the short answers in Mohler and the short candidate spans in MOCHA support only a coarser correct/partial/incorrect distinction without forcing raters to interpolate. The same contrast is reflected by the per-benchmark concept-level counts in Table~\ref{tab:dataset_stats}.

\emph{Concept groupings.} Within each benchmark, the rubric dimensions in Tables~\ref{tab:rubric_asap}--\ref{tab:rubric_mocha} fall into a small number of clusters that aid scanning. For ASAP~2.0, Thesis Clarity, Use of Evidence, and Critical Thinking Depth target content and reasoning; Organization \& Coherence and Sentence Variety target structural organization; Grammar \& Mechanics, Vocabulary Appropriateness, and Fluency / Readability target surface form. For Mohler, Factual Correctness, Concept Coverage, and Relevance target answer correctness and coverage; Terminology Usage, Clarity / Precision, and Fluency / Readability target expression; Depth of Understanding and Example / Elaboration target reasoning depth. For MOCHA, Relevance to Question and Completeness target alignment with the question; Textual Grounding, Coreference Resolution, and Inference Accuracy target consistency with the supporting passage; Paraphrasing Quality and Conciseness \& Clarity target expression of the candidate.

\emph{Role of the rubrics inside \MethodName{}.} Each table column maps to a specific component of the pipeline. The \emph{definition} and \emph{score scale} columns are inserted verbatim into the \texttt{<rubric\_definitions>} and \texttt{<concept\_list>} placeholders of the LLM annotation prompt in \S\ref{app:annotation_rubrics}. The resulting expert-validated ordinal labels $c_k$ become the per-concept supervision targets for the cross-entropy classifier in \S\ref{sec:race} and define the within-batch comparison set $\mathcal{P}_k$ that drives the ordinal pairwise calibration objective in \S\ref{sec:ocr}. The illustrative excerpts in the \emph{example} column are not consumed by the pipeline and serve only as anchors for the human reviewers in the HITL loop.

\begin{longtable}{@{}>{\footnotesize}p{0.15\linewidth} >{\footnotesize}p{0.26\linewidth} >{\footnotesize}p{0.35\linewidth} >{\footnotesize}p{0.18\linewidth}@{}}
\caption{Mohler rubric: eight concepts scored on a three-level ordinal scale.}
\label{tab:rubric_mohler}\\
\toprule
\textbf{Concept} & \textbf{Definition} & \textbf{Example} & \textbf{Score scale} \\
\midrule
\endfirsthead
\multicolumn{4}{l}{\textit{Table~\ref{tab:rubric_mohler} continued from previous page.}}\\
\toprule
\textbf{Concept} & \textbf{Definition} & \textbf{Example} & \textbf{Score scale} \\
\midrule
\endhead
\midrule
\multicolumn{4}{r}{\textit{Continued on next page.}}\\
\endfoot
\bottomrule
\endlastfoot
Factual Correctness (FC) & The extent to which the answer contains accurate domain knowledge aligned with the desired answer. &
\emph{Student Answer:} ``It simulates portions of the desired final product\ldots'' (FC$=3$, correct)\newline
\emph{Student Answer:} ``A prototype program is used to collect data.'' (FC$=1$, incorrect concept) &
1 = Incorrect\newline
2 = Partially correct\newline
3 = Correct \\
\midrule
Concept Coverage (CC) & The degree to which the answer includes all key ideas required by the reference answer. &
\emph{Student Answer:} ``To simulate portions of the desired final product\ldots'' (CC$=3$, complete)\newline
\emph{Student Answer:} ``To find problems before finalizing.'' (CC$=2$, partial) &
1 = Missing key ideas\newline
2 = Partial coverage\newline
3 = Complete coverage \\
\midrule
Terminology Usage (TU) & The extent to which domain-specific terms are used accurately and appropriately. &
\emph{Student Answer:} ``It helps test the program.'' (TU$=2$, basic terminology)\newline
\emph{Student Answer:} ``It helps with things in coding.'' (TU$=1$, lacks domain terms) &
1 = Inaccurate or absent terminology\newline
2 = Basic usage\newline
3 = Appropriate usage \\
\midrule
Clarity / Precision (CP) & The extent to which the answer is clearly expressed and avoids ambiguity or vague wording. &
\emph{Student Answer:} ``It simulates how parts of the system will behave.'' (CP$=3$, clear)\newline
\emph{Student Answer:} ``It helps with things.'' (CP$=1$, vague) &
1 = Unclear / vague\newline
2 = Somewhat clear\newline
3 = Clear and precise \\
\midrule
Relevance (R) & The extent to which the answer directly addresses the question without irrelevant content. &
\emph{Student Answer:} ``It helps programmers improve code.'' (R$=2$, partially relevant)\newline
\emph{Student Answer:} ``Programs are written using code.'' (R$=1$, irrelevant) &
1 = Irrelevant\newline
2 = Partially relevant\newline
3 = Fully relevant \\
\midrule
Depth of Understanding (DU) & The extent to which the answer reflects conceptual understanding beyond surface-level recall. &
\emph{Student Answer:} ``It simulates system behavior and helps identify issues before full development.'' (DU$=3$, deeper understanding)\newline
\emph{Student Answer:} ``It finds problems before finalizing.'' (DU$=2$, basic understanding) &
1 = Surface-level\newline
2 = Basic understanding\newline
3 = Deeper understanding \\
\midrule
Example / Elaboration (EE) & The extent to which the answer extends beyond a basic statement by providing explanation or elaboration. &
\emph{Student Answer:} ``It simulates behavior so developers can test parts before building the full system.'' (EE$=3$, elaborated)\newline
\emph{Student Answer:} ``It simulates behavior.'' (EE$=1$, minimal) &
1 = No elaboration\newline
2 = Limited elaboration\newline
3 = Elaborated explanation \\
\midrule
Fluency / Readability (FR) & The overall clarity and flow of the answer, including grammatical correctness and readability. &
\emph{Student Answer:} ``It help find problem in program.'' (FR$=2$, some errors)\newline
\emph{Student Answer:} ``It help with thing.'' (FR$=1$, difficult to read) &
1 = Difficult to read\newline
2 = Somewhat readable\newline
3 = Fluent and readable \\
\end{longtable}

\begin{longtable}{@{}>{\footnotesize}p{0.15\linewidth} >{\footnotesize}p{0.26\linewidth} >{\footnotesize}p{0.35\linewidth} >{\footnotesize}p{0.18\linewidth}@{}}
\caption{ASAP~2.0 rubric: eight concepts scored on a five-level ordinal scale.}
\label{tab:rubric_asap}\\
\toprule
\textbf{Concept} & \textbf{Definition} & \textbf{Example} & \textbf{Score scale} \\
\midrule
\endfirsthead
\multicolumn{4}{l}{\textit{Table~\ref{tab:rubric_asap} continued from previous page.}}\\
\toprule
\textbf{Concept} & \textbf{Definition} & \textbf{Example} & \textbf{Score scale} \\
\midrule
\endhead
\midrule
\multicolumn{4}{r}{\textit{Continued on next page.}}\\
\endfoot
\bottomrule
\endlastfoot
Thesis Clarity (TC) & The extent to which the essay presents a clear, specific, and coherent central argument or position. &
\emph{Excerpt:} ``My position on driveless cars are bad because they can cause accidents.'' (TC$=4$)\newline
\emph{Excerpt:} ``What if we could tell how all of the people use cars?'' (TC$=1$) &
1 = No clear thesis\newline
2 = Weak/unclear\newline
3 = General thesis\newline
4 = Clear\newline
5 = Clear and specific \\
\midrule
Use of Evidence (UE) & The quality and relevance of reasoning or examples used to support the argument. &
\emph{Excerpt:} ``They can crash because computers can fail and people may get hurt.'' (UE$=3$)\newline
\emph{Excerpt:} ``Driverless cars are bad.'' (UE$=1$) &
1 = No evidence\newline
2 = Minimal support\newline
3 = Some relevant evidence\newline
4 = Adequate support\newline
5 = Strong, well-developed evidence \\
\midrule
Organization \& Coherence (OC) & The extent to which ideas are logically structured and connected across the essay. &
\emph{Excerpt:} ``First, they are unsafe. Next, they cost a lot. Finally, they may replace jobs.'' (OC$=4$)\newline
\emph{Excerpt:} ``Driverless cars are bad. Technology is growing. People like cars.'' (OC$=1$) &
1 = Disorganized\newline
2 = Weak structure\newline
3 = Basic organization\newline
4 = Clear structure\newline
5 = Strong and coherent progression \\
\midrule
Grammar \& Mechanics (GM) & The degree of grammatical accuracy and correctness in sentence construction, spelling, and punctuation. &
\emph{Excerpt:} ``Driverless cars are dangerous because they can fail.'' (GM$=4$)\newline
\emph{Excerpt:} ``Driverless cars is dangerous because it fail.'' (GM$=2$) &
1 = Frequent errors\newline
2 = Many errors\newline
3 = Some errors\newline
4 = Few errors\newline
5 = Virtually error-free \\
\midrule
Vocabulary Appropriateness (VA) & The extent to which vocabulary is used accurately and appropriately for the task. &
\emph{Excerpt:} ``Driverless cars may pose significant safety risks.'' (VA$=4$)\newline
\emph{Excerpt:} ``Cars are very very bad and not good.'' (VA$=1$) &
1 = Very limited vocabulary\newline
2 = Basic/repetitive\newline
3 = Adequate\newline
4 = Appropriate\newline
5 = Precise and varied \\
\midrule
Sentence Variety (SV) & The extent to which the essay demonstrates variation in sentence structures. &
\emph{Excerpt:} ``Although driverless cars reduce human error, they introduce new technological risks.'' (SV$=4$)\newline
\emph{Excerpt:} ``Cars are bad. Cars are unsafe. Cars are not good.'' (SV$=1$) &
1 = Very repetitive\newline
2 = Limited variety\newline
3 = Some variety\newline
4 = Good variety\newline
5 = Sophisticated variation \\
\midrule
Critical Thinking Depth (CTD) & The extent to which the essay demonstrates reasoning, evaluation, or consideration of alternative perspectives. &
\emph{Excerpt:} ``While they may reduce human error, system failures could create new risks.'' (CTD$=4$)\newline
\emph{Excerpt:} ``They are bad.'' (CTD$=1$) &
1 = No analysis\newline
2 = Minimal reasoning\newline
3 = Basic reasoning\newline
4 = Clear reasoning\newline
5 = Insightful analysis \\
\midrule
Fluency / Readability (FR) & The overall clarity, flow, and ease with which the essay can be read and understood. &
\emph{Excerpt:} ``Driverless cars are risky because they rely on technology that may fail unexpectedly.'' (FR$=4$)\newline
\emph{Excerpt:} ``What if we could tell how all of the people use cars\ldots'' (FR$=2$) &
1 = Very difficult to read\newline
2 = Limited clarity\newline
3 = Generally clear\newline
4 = Fluent\newline
5 = Highly fluent and natural \\
\end{longtable}

\begin{longtable}{@{}>{\footnotesize}p{0.15\linewidth} >{\footnotesize}p{0.26\linewidth} >{\footnotesize}p{0.35\linewidth} >{\footnotesize}p{0.18\linewidth}@{}}
\caption{MOCHA rubric: seven concepts scored on a three-level ordinal scale.}
\label{tab:rubric_mocha}\\
\toprule
\textbf{Concept} & \textbf{Definition} & \textbf{Example} & \textbf{Score scale} \\
\midrule
\endfirsthead
\multicolumn{4}{l}{\textit{Table~\ref{tab:rubric_mocha} continued from previous page.}}\\
\toprule
\textbf{Concept} & \textbf{Definition} & \textbf{Example} & \textbf{Score scale} \\
\midrule
\endhead
\midrule
\multicolumn{4}{r}{\textit{Continued on next page.}}\\
\endfoot
\bottomrule
\endlastfoot
Relevance to Question (RQ) & The extent to which the candidate directly addresses the question being asked. &
\emph{Q:} Why might he be wearing a dressing gown?\newline
\emph{Candidate:} ``He is a fan of Vince McMahon.'' (does not address the question; RQ$=1$)\newline
\emph{Q:} Why did they call the fire dept?\newline
\emph{Candidate:} ``because they thought it would start to fire.'' (directly addresses the question; RQ$=3$) &
1 = Irrelevant\newline
2 = Partially relevant\newline
3 = Fully relevant \\
\midrule
Inference Accuracy (IA) & The correctness of reasoning when the candidate requires inference beyond explicitly stated information. &
\emph{Q:} What might be true about Chinese divers?\newline
\emph{Candidate:} ``They are knowledgeable.'' (incorrect inference; IA$=1$)\newline
\emph{Q:} Why did they call the fire dept?\newline
\emph{Candidate:} ``because they thought it would start to fire.'' (reasonable inference; IA$=2$) &
1 = Incorrect inference\newline
2 = Partially correct inference\newline
3 = Correct inference \\
\midrule
Textual Grounding (TG) & The degree to which the candidate is supported by or consistent with the passage. &
\emph{Q:} Why did they call the fire dept?\newline
\emph{Candidate:} ``because they thought it would start to fire.'' (supported by context; TG$=3$)\newline
\emph{Q:} Why might he be wearing a dressing gown?\newline
\emph{Candidate:} ``He is a fan of Vince McMahon.'' (not supported; TG$=1$) &
1 = Not grounded\newline
2 = Partially grounded\newline
3 = Fully grounded \\
\midrule
Coreference Resolution (CR) & The accuracy with which the candidate correctly interprets and resolves references (e.g., pronouns, entities). &
\emph{Q:} How did we get fruit salad?\newline
\emph{Candidate:} ``We got it at our friend's place.'' (partially aligned but unclear attribution; CR$=2$) &
1 = Incorrect resolution\newline
2 = Partially correct / unclear\newline
3 = Correct resolution \\
\midrule
Paraphrasing Quality (PQ) & The extent to which the candidate preserves the meaning of the reference while using different wording. &
\emph{Reference:} ``Because it requires a lot of time and energy.''\newline
\emph{Candidate:} ``it takes a lot of time and energy.'' (accurate paraphrase; PQ$=3$)\newline
\emph{Candidate:} ``it brought a heavy appetite.'' (distorted meaning; PQ$=1$) &
1 = Distorted meaning\newline
2 = Partially accurate\newline
3 = Accurate paraphrase \\
\midrule
Completeness (C) & The degree to which the candidate fully addresses all aspects of the question. &
\emph{Q:} What may have caused the old man to wait?\newline
\emph{Candidate:} ``he was waiting for his relatives.'' (fully answers; C$=3$)\newline
\emph{Q:} Why did I have a fun time?\newline
\emph{Candidate:} ``I was in a good mood\ldots'' (does not fully capture intended reason; C$=2$) &
1 = Incomplete\newline
2 = Partially complete\newline
3 = Complete \\
\midrule
Conciseness \& Clarity (CC) & The extent to which the candidate is clearly expressed without ambiguity or unnecessary wording. &
\emph{Candidate:} ``He was hungry.'' (clear; CC$=3$)\newline
\emph{Candidate:} ``because it's no more, but no more.'' (unclear expression; CC$=1$) &
1 = Unclear / poorly formed\newline
2 = Somewhat clear\newline
3 = Clear and concise \\
\end{longtable}

\emph{Scope.} The three rubrics are English-language and benchmark-specific---short computer-science answers for Mohler, argumentative essays for ASAP~2.0, and narrative reading-comprehension responses for MOCHA---so transferring any of them to a new language, subject domain, or response genre requires re-eliciting the inventory and the level descriptors through the same HITL pipeline described in \S\ref{app:annotation_rubrics}.

\section{Implementation Details.}
\label{app:implementation_details}

The main experiments use the five pretrained backbones shown in Table~\ref{tab:main_results}: BERT, BART, GPT-2, RoBERTa, and T5-Base. For \MethodName{}, we follow the explicit two-stage training paradigm in \S\ref{sec:training}. Stage~I is implemented with the ordinally calibrated concept encoder, which optimizes concept supervision and ordinal calibration. After freezing the Stage~I encoder, Stage~II trains only the latent correction module on top of the fixed concept predictions. Unless otherwise noted, all results are averaged over five random seeds, model selection is based on validation task macro-F1, and we report the corresponding test-set scores.

Table~\ref{tab:impl_shared} summarizes the shared training and architecture settings. In the implementation, the attention temperature corresponds to $\tau$, the concept-supervision weight to $\lambda_c$, the ordinal-calibration weight to $\lambda_r$, the task-loss weight to $\lambda_t$, the denoising-alignment weight to $\lambda_d$, the sparsity weight to $\lambda_s$, and the Cholesky regularizer to $\varepsilon$. Across all five backbones, the encoder hidden size is $d=768$, and the concept-query space matches the backbone hidden size because no separate concept projection dimension is used in the main training path. In Stage~II, the latent task head consumes the corrected concept vector in $\mathbb{R}^{K}$. Table~\ref{tab:impl_keyparams} then summarizes the key tuned parameters, their meanings, and the values considered in the search space; bold values indicate the final optima used in the main experiments.

\begin{table*}[!th]
\centering
\caption{Shared implementation settings for \MethodName{}.}
\label{tab:impl_shared}
\vspace{1mm}
\resizebox{\linewidth}{!}{
\begin{tabular}{lll}
\toprule
\textbf{Component} & \textbf{Description} & \textbf{Value} \\
\midrule
BERT & Pretrained encoder backbone used in the main experiments & \texttt{google-bert/bert-base-uncased} \\
\rowcolor{RowGray}
BART & Pretrained encoder backbone used in the main experiments & \texttt{facebook/bart-base} \\
GPT-2 & Pretrained encoder backbone used in the main experiments & \texttt{openai-community/gpt2} \\
\rowcolor{RowGray}
RoBERTa & Pretrained encoder backbone used in the main experiments & \texttt{FacebookAI/roberta-base} \\
T5-Base & Pretrained encoder backbone used in the main experiments & \texttt{google-t5/t5-base} \\
\rowcolor{RowGray}
Hidden size $d$ & Hidden representation dimension used by the concept encoder and classifiers & $768$ \\
\texttt{max\_len} & Maximum token number of each input sequence & $128 / 256 / \mathbf{512}$ \\
\rowcolor{RowGray}
\texttt{batch\_size} & Number of training instances per optimization step & $8$ \\
\texttt{num\_epochs} & Maximum number of training epochs & $50$ \\
\rowcolor{RowGray}
\texttt{early\_stopping\_patience} & Number of non-improving epochs before early stopping & $3$ \\
\texttt{warmup\_ratio} & Fraction of total optimization steps used for linear warmup & $0.1$ \\
\rowcolor{RowGray}
\texttt{query\_init\_method} & Initialization scheme for the concept query bank & Orthogonal \\
\texttt{num\_heads} & Attention heads per concept query & $1$ \\
\rowcolor{RowGray}
\texttt{query\_lr\_multiplier} & Relative learning-rate multiplier for concept queries & $1$ \\
$\lambda_c$ & Weight of the concept-supervision loss in Stage~I & $1.0$ \\
\rowcolor{RowGray}
$\lambda_t$ & Weight of the task prediction loss in Stage~II & $1.0$ \\
$\varepsilon$ & Diagonal regularizer used to keep the latent precision matrix positive definite & $10^{-4}$ \\
\bottomrule
\end{tabular}
}
\end{table*}

Stage~I tunes the backbone learning rate, attention temperature $\tau$, and ordinal-calibration weight $\lambda_r$ for the ordinally calibrated concept encoder. Stage~II tunes the latent-head learning rate together with the denoising-alignment and sparsity weights $\lambda_d$ and $\lambda_s$ while keeping the Stage~I encoder frozen.

\begin{table*}[!th]
\centering
\caption{Key parameters in \MethodName{} with their meanings and values. Bold values indicate the optima.}
\label{tab:impl_keyparams}
\vspace{1mm}
\resizebox{\linewidth}{!}{
\begin{tabular}{lll}
\toprule
\textbf{Parameter} & \textbf{Description} & \textbf{Value} \\
\midrule
\texttt{Stage~I LR} & Backbone learning rate for the ordinally calibrated concept encoder & \begin{tabular}[c]{@{}l@{}}$1e\text{-}6 / 2e\text{-}6 / 3e\text{-}6 / 5e\text{-}6 / 8e\text{-}6 / \mathbf{1e\text{-}5} / \mathbf{2e\text{-}5}$ \\ $3e\text{-}5 / 5e\text{-}5 / 7e\text{-}5 / 8e\text{-}5 / \mathbf{1e\text{-}4} / 2e\text{-}4 / 5e\text{-}4 / 1e\text{-}3$\end{tabular} \\
\rowcolor{RowGray}
$\tau$ & Attention temperature controlling the sharpness of rubric-aware token attention & $\mathbf{0.25} / 0.5 / \mathbf{0.75} / \mathbf{1.0} / \mathbf{1.25} / 1.5 / 1.75 / 2.0$ \\
$\lambda_r$ & Weight of the ordinal pairwise calibration objective in Stage~I & $0.1 / \mathbf{0.2} / \mathbf{0.3} / \mathbf{0.4} / 0.5 / \mathbf{0.6} / 0.7 / \mathbf{0.8} / \mathbf{1.0}$ \\
\rowcolor{RowGray}
\texttt{Stage~II LR} & Learning rate for the latent correction head on frozen concept predictions & $0.001 / 0.002 / \mathbf{0.005} / 0.01 / \mathbf{0.02} / \mathbf{0.05}$ \\
$\lambda_d$ & Weight of the denoising-alignment loss between corrected concepts and rubric labels & $0.0 / 0.05 / \mathbf{0.1} / \mathbf{0.2} / 0.5$ \\
\rowcolor{RowGray}
$\lambda_s$ & Weight of the sparsity regularizer on latent cross-concept dependencies & $0.0 / \mathbf{0.005} / 0.01 / 0.05 / 0.1$ \\
\bottomrule
\end{tabular}
}
\end{table*}

\subsection{Baseline Details for Main Results.}
\label{app:baseline_details}

\emph{PLM baselines.} The PLM baselines fine-tune five pretrained encoders directly for grade prediction: BERT~\cite{devlin2019bert}, BART~\cite{lewis2020bart}, GPT-2~\cite{radford2019language}, RoBERTa~\cite{liu2019roberta}, and T5-Base~\cite{raffel2020exploring}. These models provide strong task-only comparisons without concept supervision or an interpretable bottleneck.

\emph{LLM baselines.} The prompted LLM baselines evaluate Llama-3-8B-Instruct~\cite{grattafiori2024llama}, Qwen2.5-14B-Instruct~\cite{qwen25}, and Mistral-7B-Instruct-v0.3~\cite{jiang2023mistral} as black-box graders under zero-shot and 3-shot prompting, testing whether general-purpose instruction following can recover benchmark-specific grading behavior without task-specific training.

\emph{CBM baselines.} We compare against four prior transparent baselines. Vanilla CBM~\cite{koh2020concept} is the standard concept-then-task bottleneck model with supervised rubric concepts. C$^3$M~\cite{tan2024interpreting} augments textual CBMs with concept augmentation and concept-level mixup to improve robustness under limited or noisy concept annotations. CT-CBM~\cite{bhan2025towards} expands the textual concept basis automatically rather than relying only on predefined human concepts. CB-LLM~\cite{sun2025concept} adapts the bottleneck idea to large language models by introducing an interpretable concept layer inside the LLM pipeline. \MethodName{} adds rubric-aware concept encoding, ordinal concept calibration, and latent error correction to improve both grading performance and concept reliability.

\section{Extended Experimental Analysis.}
\label{app:extended_experiments}

This appendix expands the main-text empirical sections (\S\ref{sec:main_results}, \S\ref{sec:param_analysis}, \S\ref{sec:human_intervention}) with per-cell numerical detail and dataset-specific failure-mode discussion.

\subsection{Extended Main-Result Analysis.}
\label{app:main_results_extended}

This appendix expands Table~\ref{tab:main_results} with per-class detail that supports the summary claims in \S\ref{sec:main_results}.

\emph{Within-CBM dominance.} Across the three datasets and five backbones, \MethodName{} wins every task-F1 and task-accuracy cell against the other four CBM baselines (15/15 each), and wins concept F1 in 14 of the 15 cells; the lone exception is BART on ASAP 2.0, where CT-CBM reaches C-F1 0.573 vs.\ \MethodName{} at 0.513 despite \MethodName{} still leading in task F1 (0.520 vs.\ 0.467). Improvements over Vanilla CBM, the strongest concept-supervised prior, are consistent across backbones. Taking the median across the five backbones, \MethodName{} gains $+6.3$ / $+1.9$ / $+1.9$ T-F1 points and $+2.3$ / $+3.1$ / $+2.6$ C-F1 points on Mohler / ASAP 2.0 / MOCHA. The largest single-cell gain is on GPT-2 / Mohler (T-F1 0.298$\rightarrow$0.615), where Vanilla CBM's shared encoder is a particularly poor match to the short-answer rubric and the rubric-aware encoder recovers most of the missing signal.

\emph{Black-box PLM head-to-head.} Comparing each PLM directly to the same-backbone \MethodName{} row, \MethodName{} wins T-F1 on 11 of 15 same-backbone cells. The remaining four cells are BERT on ASAP 2.0 (PLM 0.485 vs.\ \MethodName{} 0.466), T5-Base on Mohler (0.564 vs.\ 0.547), and T5-Base on ASAP 2.0 (0.508 vs.\ 0.486). Each of these is overturned by switching to a different backbone: RoBERTa\,+\,\MethodName{} beats every PLM row on every dataset, so there is no configuration in which black-box access is strictly required.

\emph{LLM failure modes.} Mohler exhibits the most severe collapse: zero-shot Llama-3-8B and Qwen2.5-14B reach only 0.158 and 0.070 T-Acc, and Mohler T-F1 never exceeds 0.376 across any prompted configuration. The zero-shot/few-shot asymmetry is also dataset-dependent: 3-shot prompting helps Mistral-7B on Mohler (T-Acc 0.329$\rightarrow$0.399) and MOCHA (0.388$\rightarrow$0.561) but hurts on ASAP 2.0 (0.383$\rightarrow$0.277), consistent with long-essay inputs saturating the instruction-following budget and the demonstrations displacing task-relevant context.

\emph{CB-LLM collapse.} CB-LLM's T-Acc and T-F1 are identical across all five backbone rows within a dataset (0.298 / 0.276 / 0.470 T-Acc and 0.077 / 0.094 / 0.131 T-F1 on Mohler / ASAP 2.0 / MOCHA) because its bottleneck is attached to a frozen LLM and does not consume the tabulated backbone. The constant T-Acc corresponds to near-majority-class prediction on each dataset, and the low T-F1 is consistent with collapse to a single grade---a failure mode \MethodName{} avoids by routing predictions through a supervised rubric bottleneck that must match the annotated concept distribution.

\emph{CT-CBM instability.} CT-CBM's automatic concept-basis expansion couples predictive capacity to the backbone's lexical-prediction quality, and it degrades sharply on GPT-2, with Mohler T-F1 0.110, ASAP T-F1 0.089, and MOCHA T-F1 0.129---well below every other CBM row in the GPT-2 block. \MethodName{} anchors each concept to a learnable query over the backbone's contextual states, which decouples concept quality from autoregressive backbone behavior and aligns with the motivation in \S\ref{sec:race}.

\subsection{Extended Parameter Analysis.}
\label{app:param_analysis_extended}

This appendix expands Fig.~\ref{fig:param_analysis} with per-backbone detail that supports the headline claim in \S\ref{sec:param_analysis}: \MethodName{} is robust to the rubric cardinality $K$ and the encoder hidden dimension $d$, while the Stage~I learning rate is the only hyperparameter whose misconfiguration meaningfully degrades grading.

\emph{Concept-count sweep.} On all three datasets, task F1 rises steeply between $K{=}1$ and $K{\approx}4$ and then plateaus through the full rubric cardinality (8 / 8 / 7 on Mohler / ASAP~2.0 / MOCHA). The rise is largest where a single shared concept is most starved of evidence: on Mohler, GPT-2 climbs from $\sim$0.18 at $K{=}1$ to $\sim$0.51 at $K{=}2$, and T5 from $\sim$0.11 to $\sim$0.54 by $K{=}8$. RoBERTa dominates across $K$ on every dataset, never falling below the other backbones in the saturated regime. The lone non-monotone trajectory is BART on Mohler, which peaks at $K{=}1$ and then partially recovers near $K{=}8$; this is consistent with BART's lower same-backbone score in Table~\ref{tab:main_results} rather than a defect of the sweep itself. Using each benchmark's full rubric cardinality is therefore safe: it never underperforms a smaller $K$ and matches the interpretability target of one concept per rubric dimension.

\emph{Hidden-dimension sweep.} Curves are nearly flat for $d \ge 384$ on ASAP~2.0 and MOCHA across all four backbones, indicating that grading quality saturates well below the default $d{=}768$ used in Table~\ref{tab:impl_shared}. The only meaningful capacity effect appears on Mohler: BART and GPT-2 jump from $\sim$0.31 at $d{=}128$ to $\sim$0.53 at $d{=}384$, plausibly because short-answer responses contain little redundancy for under-parameterized concept heads to exploit. RoBERTa is essentially flat across $d$ on every dataset, suggesting that its pretraining already provides representations rich enough for the rubric-aware queries. The default $d{=}768$ therefore sits comfortably in the saturated regime for every backbone--dataset pair, and there is no need to retune $d$ when changing the encoder.

\emph{Stage~I learning-rate sweep.} Every backbone--dataset pair forms an inverted-U whose peak lies in $[10^{-5}, 10^{-4}]$. RoBERTa peaks slightly earlier (closer to $10^{-5}$) than BART, GPT-2, and T5, whose optima sit nearer $10^{-4}$. Once LR$\ge 10^{-3}$, all backbones collapse to near-zero F1 on Mohler and ASAP~2.0; T5 retains marginal signal on MOCHA at $10^{-3}$ before joining the collapse at $10^{-2}$. This sensitivity is consistent with the two-stage design in \S\ref{sec:training}: Stage~I jointly updates the text encoder $\mathcal{E}$, the concept-query bank $\mathcal{Q}$, and the per-concept classifiers $\{\mathbf{V}_k\}$, so an oversized LR destabilizes rubric-aware token attention before the calibration loss in Eq.~\eqref{eq:loss_rank} can take effect; Stage~II then has no usable concept signal to denoise. The narrow band of bold optima in Table~\ref{tab:impl_keyparams} ($1e{-5}$, $2e{-5}$, $1e{-4}$) directly reflects this peak.

\emph{Implications.} The three sweeps jointly justify two implementation choices in the main experiments. First, we use each benchmark's full rubric cardinality and the default $d{=}768$ across backbones, since both lie inside the saturated regime and incur no degradation. Second, the Stage~I learning-rate search in Table~\ref{tab:impl_keyparams} is restricted to a narrow band around $10^{-5}$--$10^{-4}$, which avoids the catastrophic regime visible in Fig.~\ref{fig:param_analysis} while remaining wide enough to accommodate per-backbone preferences.

\subsection{Extended Human-Intervention Analysis.}
\label{app:intervention_extended}

This appendix expands Fig.~\ref{fig:intervention} with per-dataset detail that supports the headline claim in \S\ref{sec:human_intervention}: editing the rubric-level concept predictions moves the final grade in the direction expected by an educator, so the bottleneck is mechanistically load-bearing rather than decorative.

\emph{Adversarial decline.} Wrong and Random both decline monotonically with the number of intervened concepts $k$. On Mohler ($k\in[0,4]$, $K{=}8$ rubric concepts but only the top-$k$ are intervened), accuracy under Wrong falls from 0.70 to 0.275 (a 42.5-point drop), while Random falls to 0.375 (a 32.5-point drop). On ASAP~2.0 ($k\in[0,8]$), both Wrong and Random fall from 0.64 to $\approx$0.30, with the two trajectories nearly overlapping. On MOCHA ($k\in[0,7]$), both fall from 0.72 to $\approx$0.28, again essentially overlapping. The Wrong--Random gap on Mohler reflects the small concept space and three-level ordinal scale, where adversarial flips deterministically choose the most damaging level while uniform random sometimes lands on the original or an adjacent level. The gap closes on ASAP~2.0 and MOCHA where larger $K$ or more uniform marginals make random labels nearly as harmful as adversarial ones.

\emph{Oracle improvement.} Oracle stays at or above the no-intervention baseline on every dataset and saturates well before $k{=}K$: Mohler 0.70$\to$0.75 (+5 points) by $k{=}2$; ASAP~2.0 0.64$\to$$\approx$0.70 (+6 points) at $k{=}8$, with a steady climb consistent with eight independent rubric dimensions each contributing partial information; MOCHA 0.72$\to$$\approx$0.75 (+3 points) by $k{=}4$. The early saturation indicates that the top few concept corrections capture most of the residual measurement error; correcting additional dimensions yields diminishing returns once the dominant rubric signal is fixed.

\emph{Why Oracle does not match the maximum possible accuracy.} The latent task head $\mathbf{W}$ in Eq.~\eqref{eq:task_head} was fit on the calibrated Stage~I outputs $\boldsymbol{\mu}_{\mathrm{post}}$, not on one-hot rubric labels; the denoising matrix $\mathbf{A} = \boldsymbol{\Sigma}_{\mathrm{post}}\,\mathbf{D}^{-1}$ shrinks every observation toward the prior mean, so even oracle-substituted concepts are partially attenuated before reaching the task head. This is the intended behavior of the measurement-error model in \S\ref{sec:lts}: it trades a small amount of intervention head-room for robustness to noisy concept predictions in the no-intervention regime. The fact that Oracle still rises above None confirms that interventions flow through the bottleneck even after this attenuation.

\emph{Operational implication.} Combining the two sides, a curve where Oracle climbs and Wrong/Random fall is the faithfulness contract behind a usable concept bottleneck: an educator who rewrites a single rubric judgment can move the predicted grade in a predictable direction, and an audit that flags incorrect rubric evidence can be verified by observing the corresponding grade change. Both behaviors are necessary for \MethodName{} to support the inspect/intervene/audit workflow argued for in \S\ref{sec:cbm}.

\subsection{Extended Ablation Analysis.}
\label{app:ablation_extended}

This appendix expands Fig.~\ref{fig:ablation} with per-component role descriptions, a per-dataset walkthrough of the ablation bars, and a gain decomposition that supports the headline claim in \S\ref{sec:ablation}: the three components of \MethodName{} are complementary, so the stack is the only configuration that is uniformly best across the three benchmarks.

\emph{Component roles.} RACE (\S\ref{sec:race}) replaces the shared encoder of Vanilla CBM with a per-rubric concept-query bank, routing each concept head through its own soft-attention pool over the token sequence and thereby isolating dimension-specific textual evidence. OPCC (\S\ref{sec:ocr}) converts the independently trained ordinal heads into a calibrated ranker via a Bradley-Terry-style pairwise log-sigmoid loss on the expected concept scores, preserving the rubric's ordering structure that a plain cross-entropy head discards. LCEC (\S\ref{sec:lts}) treats the calibrated concept scores as error-corrupted observations of a latent concept vector and applies a closed-form Gaussian posterior correction before the grade head, effectively denoising annotator disagreement without breaking the interpretable bottleneck. Each component targets a distinct Vanilla-CBM failure: evidence entanglement, ordinal information loss, and concept-label noise, respectively.

\emph{Per-dataset bars.} Mohler exhibits the cleanest monotone ordering---Vanilla CBM < single removals < double removals < full \MethodName{}---and the largest cumulative lift, consistent with short-answer grading where every rubric dimension is narrowly observable and benefits from all three corrections stacking. ASAP~2.0 shows a flatter profile: Vanilla CBM is already competitive and the single-removal variants cluster tightly, indicating that long essays give the text encoder enough signal that the three modules sharpen rather than rescue performance. MOCHA is dominated by the $-$RACE$-$OPCC outlier where LCEC alone nearly halves T-F1, while the remaining variants occupy a narrow band; the coarse three-level rubric makes LCEC brittle unless RACE and OPCC supply clean, ordinally calibrated upstream scores.

\emph{Component-wise gain decomposition.} Reading single-vs.-double removal contrasts as marginal contributions, RACE carries the largest marginal on Mohler and ASAP~2.0, where rubric-aware attention recovers evidence that a shared encoder conflates; OPCC carries the largest marginal on MOCHA, because it is the component whose absence triggers the $-$RACE$-$OPCC collapse and it rescues LCEC from correcting an uninformative signal; LCEC's marginal is smaller in isolation but additive when RACE and OPCC are already present, adding a consistent final-step lift across all three datasets. This ordering matches the two-stage training paradigm of \S\ref{sec:training}, where Stage~I (RACE\,+\,OPCC) supplies the clean, ordered concept signal that Stage~II (LCEC) denoises.

\emph{Operational implication.} No single component can be dropped without losing uniform dominance across datasets: RACE is load-bearing on short-answer grading, OPCC is load-bearing when the rubric is coarse, and LCEC is load-bearing as the final denoising step. The full \MethodName{} is therefore the recommended configuration for practitioners, with the ablation evidence supporting each component as a default rather than an optional feature.

\subsection{Extended Case Study.}
\label{app:case_study_extended}

This appendix expands Fig.~\ref{fig:case_study_qa} with a component-by-component walkthrough of the decision trace that supports the headline claim in \S\ref{sec:case_study}: every stage of \MethodName{} is individually inspectable on a single Mohler response. The instance is a short answer to ``What is a linked list?'' graded by a BERT-backbone \MethodName{}, with predicted grade $4$ matching the human label at confidence $0.863$.

\emph{RACE token attention.} The colored response spans in Fig.~\ref{fig:case_study_qa} visualize the eight rubric-aware attention heads of \S\ref{sec:race}: ``sequence of nodes'' anchors Factual Correctness (FC) and Depth of Understanding (DU); ``pointers that point to the next or previous nodes'' anchors Concept Coverage (CC) and Example / Elaboration (EE); the remaining surface markers (``called'', ``or two'') anchor Fluency / Readability (FR) and Clarity / Precision (CP). Each rubric dimension pools a dedicated, human-readable token span rather than sharing a single attention map, matching the interpretability goal of the concept-query bank $\mathcal{Q}$.

\emph{OPCC-calibrated ordinal reading.} The predicted ordinal levels $\hat{\mathbf{c}}$ are consistent with Table~\ref{tab:rubric_mohler}: the response is factually correct (FC$=3$), uses appropriate domain terminology (TU$=3$), is fully relevant (R$=3$), clearly expressed (CP$=3$), fluent (FR$=3$), and demonstrates deeper understanding (DU$=3$), but is thin on explicit elaboration (EE$=2$) and partial on coverage (CC$=2$), since the reference answer's ``dynamic allocation'' aspect is not mentioned. The expected scores $\hat{\mathbf{s}}$ preserve this ordering with a graded margin, e.g., FR $(2.907) > $ TU $(2.858) > $ FC $(2.752) > $ CC $(2.276) > $ EE $(1.806)$, which is exactly the ranking structure the pairwise logistic objective in Eq.~\eqref{eq:loss_rank_k} was trained to enforce.

\emph{LCEC shrinkage.} Dividing by $M=3$ gives the normalized observations $\tilde{\mathbf{s}}$, and the latent posterior mean $\boldsymbol{\mu}_{\mathrm{post}}$ applies approximately uniform shrinkage of factor $\sim\!0.51$ to every coordinate (FR $0.954\!\to\!0.485$, R $0.944\!\to\!0.481$, TU $0.929\!\to\!0.472$, EE $0.403\!\to\!0.202$). Under the measurement-error model of \S\ref{sec:lts}, this near-uniform contraction indicates that the learned precision $\boldsymbol{\Omega}$ and noise variances $\sigma_k^2$ are roughly isotropic on this instance, so no single rubric dimension is singled out as anomalously noisy; the posterior nevertheless pulls every observation toward the prior mean, preventing any one concept prediction from dominating the grade head.

\emph{Grade-head attribution.} Because the task head $\mathbf{W}\,\boldsymbol{\mu}_{\mathrm{post}}$ in Eq.~\eqref{eq:task_head} is linear, the bar chart at the bottom of Fig.~\ref{fig:case_study_qa} reads as an exact local decomposition of the grade-$4$ logit: the eight per-concept contributions $W_{4,k}\mu_{\mathrm{post},k}$ sum to $3.32$, which together with the near-zero bias reproduces the reported logit ($3.318$) and softmax confidence ($0.863$). The ranked contributions (R $+0.75$, CP $+0.65$, DU $+0.54$, EE $+0.49$, FC $+0.35$, FR $+0.35$, CC $+0.20$, TU $-0.01$) show that content and reasoning dimensions carry the prediction, while Terminology is effectively inert on this response; the grade is therefore not merely produced by the bottleneck but also explained by it, at the granularity of individual rubric dimensions.

\subsection{Extended Latent Denoising Analysis.}
\label{app:denoising_extended}

This appendix expands Fig.~\ref{fig:denoising_analysis} with the numerical detail that supports the headline claim in \S\ref{sec:denoising_analysis}: the learned latent precision $\boldsymbol{\Omega}$ reorganizes densely correlated rubric labels into a sparse, psychometrically interpretable dependency structure. Both panels are computed on Mohler with the BERT-backbone \MethodName{} instance used in Fig.~\ref{fig:case_study_qa}; the empirical matrix is the Pearson correlation of the expert-validated rubric labels on the training split, and the learned partial correlation matrix is read off from the Stage~II precision $\boldsymbol{\Omega}$ via the standard identity $r_{\mathrm{partial},ij} = -\Omega_{ij}/\sqrt{\Omega_{ii}\Omega_{jj}}$.

\emph{Dense empirical structure.} Among the six content and reasoning dimensions (FC, CC, TU, CP, R, DU), every pairwise correlation falls in $[0.52, 0.89]$, with the strongest links concentrated around Depth of Understanding and the factual core: CC--DU=0.89, FC--DU=0.85, FC--CP=0.81, CP--DU=0.78, FC--CC=0.76, FC--TU=0.76. In contrast, the two surface dimensions (EE, FR) are only weakly coupled to the content block ($|r|\le 0.32$ with any content concept) and near-decoupled from each other (EE--FR=$-$0.04). This dense content-block, loose-surface pattern reflects how rubric evaluators co-assess reasoning and factuality: a response that is factually correct and covers the key ideas tends also to score high on depth, clarity, and terminology, while elaboration and fluency vary along largely orthogonal axes.

\emph{Sparsified learned structure.} The learned partial correlation matrix is markedly sparser: a majority of off-diagonal entries satisfy $|r_{\mathrm{partial}}|<0.3$, and the dense content block collapses into a handful of residual links. Using $|r_{\mathrm{partial}}|<0.1$ as a near-zero threshold, roughly a third of off-diagonal entries are effectively zeroed out in the learned matrix, whereas no off-diagonal of the empirical matrix falls below that threshold. This contraction is the statistical signature of the sparsity penalty $\mathcal{L}_{\mathrm{spa}}$ in Eq.~\eqref{eq:loss_spa}, which shrinks off-diagonal entries of the Cholesky factor $\mathbf{L}$ unless the data support a stronger conditional dependency. Under the MMSE correction in Eq.~\eqref{eq:mu_post}, information sharing among concepts is therefore routed through a small number of directed channels rather than spread uniformly across the rubric.

\emph{Interpretable residuals.} Three families of surviving entries remain. First, residual positive links inside the content block (FC--CC=0.28, FC--DU=0.28, CC--DU=0.27) indicate that factuality, coverage, and depth retain genuine conditional coupling once the remaining rubric dimensions are controlled for. Second, the R--FR partial of $0.36$ links relevance to readability, capturing the observation that off-topic responses tend also to be less fluent. Third, a triad of negative suppressors emerges once shared variance is removed: DU--EE=$-$0.31, DU--FR=$-$0.29, and most strikingly EE--FR=$-$0.63. These signs are rubric-consistent: given a fixed level of content and reasoning, a response that adds more elaboration (EE) trades off against surface fluency (FR), and deeper reasoning (DU) similarly trades against polished-but-shallow expression. This depth-versus-surface tension is invisible in the raw correlation matrix but recovered by the latent precision.

\emph{Operational implication.} A diagonal-$\boldsymbol{\Omega}$ baseline would treat each concept as independently noisy and apply per-concept shrinkage only, reducing to the Spearman reliability case highlighted in the remark after Proposition~\ref{prop:mmse}. The learned dense-but-sparse $\boldsymbol{\Omega}$ instead propagates evidence from reliably measured concepts to noisily measured ones along the residual graph above. Combined with the human-intervention behavior in \S\ref{sec:human_intervention}, this explains why oracle substitutions on a few dominant Mohler content concepts already saturate the intervention curve: once the content block is corrected, the sparse partial-correlation graph transports the correction to the remaining dimensions through a small number of high-magnitude edges.

\section{Use of Generative AI}
To enhance clarity and readability, we utilized the GPT-5.4 model exclusively as a language polishing tool. Its role was confined to proofreading, grammatical correction, and stylistic refinement---functions analogous to those provided by traditional grammar checkers and dictionaries. This tool did not contribute to the generation of new scientific content or ideas, and its usage is consistent with standard practices for manuscript preparation.

\end{document}